\documentclass[a4paper, USenglish, cleveref, autoref, thm-restate]{lipics-v2019}

\input{preamble}

\title{Logic of computational semi-effects and categorical gluing for equivariant functors}

\author{Yuichi Nishiwaki}{University of Tokyo, Japan}{nyuichi@is.s.u-tokyo.ac.jp}{}{}

\author{Toshiya Asai}{University of Tokyo, Japan}{anori@is.s.u-tokyo.ac.jp}{}{}

\authorrunning{Y. Nishiwaki and T. Asai}

\Copyright{Yuichi Nishiwaki and Toshiya Asai}

\ccsdesc[500]{Theory of computation~Categorical semantics}

\keywords{computational effects, actegories, logical relations, categorical gluing, proof theoretic semantics, modal logic, Curry-Howard correspondence, fibrations}

\category{} 

\relatedversion{} 

\supplement{}

\nolinenumbers 

\hideLIPIcs  

\EventEditors{John Q. Open and Joan R. Access}
\EventNoEds{2}
\EventLongTitle{42nd Conference on Very Important Topics (CVIT 2016)}
\EventShortTitle{CVIT 2016}
\EventAcronym{CVIT}
\EventYear{2016}
\EventDate{December 24--27, 2016}
\EventLocation{Little Whinging, United Kingdom}
\EventLogo{}
\SeriesVolume{42}
\ArticleNo{23}

\begin{document}

\maketitle

\begin{abstract}
In this paper, we revisit Moggi's celebrated calculus of computational effects from the perspective of logic of monoidal action (actegory).
Our development takes the following steps.
Firstly, we perform proof-theoretic reconstruction of Moggi's computational metalanguage and obtain a type theory with a modal type $\rhd$ as a refinement.
Through the proposition-as-type paradigm, its logic can be seen as a decomposition of lax logic via Benton's adjoint calculus.
This calculus models as a programming language a weaker version of effects, which we call \emph{semi-effects}.
Secondly, we give its semantics using actegories and equivariant functors.
Compared to previous studies of effects and actegories, our approach is more general in that models are directly given by equivariant functors, which include Freyd categories (hence strong monads) as a special case.
Thirdly, we show that categorical gluing along equivariant functors is possible and derive logical predicates for $\rhd$-modality.
We also show that this gluing, under a natural assumption, gives rise to logical predicates that coincide with those derived by Katsumata's categorical $\top\top$-lifting for Moggi's metalanguage.
\end{abstract}

\section{Introduction}

It has passed about three decades since the deep connection between the notion of computation and monads in category theory was revealed by Moggi \cite{Moggi89, Moggi91}.
Moggi's papers have been not only affecting the design of most modern programming languages but also standing as a very foundation of the semantic analysis of the notion of computation.
His insight that computation can be modeled by monads is now widely accepted and sometimes even considered as ``general knowledge'' in the community.
In this paper, we revisit this prevalent slogan ``computation as monads'' with a somewhat critical eye, and attempt to propose our new alternative: ``computation as monoidal actions''.

One important contribution in Moggi's papers was the suggestion of a formal system of equational logic (called \emph{the metalanguage}) that can uniformly represent various kinds of computation where the type of involved effect is given as a parameter.
Once the papers' importance was recognized, the metalanguage began to get analyzed with the help of logic.
Lax logic \cite{DBLP:journals/iandc/FairtloughM97} (or CL logic \cite{DBLP:journals/jfp/BentonBP98}) is a modal logic with a single modality $\dia$ representing intuitionistic possibility.
Benton, Bierman, and de Paiva \cite{DBLP:journals/jfp/BentonBP98} showed that proof-term assignment to lax logic directly gives the Curry-Howard correspondence to the metalanguage.
Benton and Wadler \cite{DBLP:conf/lics/BentonW96} showed that adjoint calculus \cite{DBLP:conf/csl/Benton94} can serve as a logical foundation of the metalanguage provided that the underlying monad is commutative.
Enriched Effect Calculus (EEC) \cite{DBLP:journals/logcom/EggerMS14} pushed forward this direction.
EEC removed the limitation of the class of monads from adjoint calculus by carefully choosing the set of legitimate logical connectives and forms of typing judgments.
Along this series of works, this paper presents another reformulation of the metalanguage.
Our reformulation starts by analyzing the above logics from the proof-theoretic viewpoint, motivated by the fact that none of them enjoy \emph{stability} a la Dummett, a proof-theoretic criterion for ``nice'' logics \cite{Dummett91,DBLP:journals/jphil/Kurbis15}.
We derive our logic (and its corresponding type theory) in the following steps.
Firstly, we decompose lax logic into a logic with two adjoint modalities $\lhd$ and $\rhd$ exploiting the technique used by adjoint calculus.
Secondly, to maintain the non-commutativity of the metalanguage, we restrict the number of possible variables in the realm of computation to at most one.
At the final step, we throw away the $\lhd$ modality.
The resulted calculus is a rather weak system due to the lack of an adjoint modality.
Nevertheless, it enjoys nice proof-theoretic properties, maintains the essence of computation, and has a simple categorical model.
We call this calculus \emph{semi-effect calculus (SEC)}, after its operational behavior.

Proposal of strong monads as a categorical semantics of computation is also a significant contribution of Moggi's papers.
Along this line, many further studies have been done so far \cite{DBLP:conf/icalp/PowerT99,DBLP:journals/entcs/Fuhrmann99,DBLP:journals/entcs/Staton14}.
For example, M{\o}gelberg and Staton \cite{MogelbergAndStaton} used copowers in enriched categories to interpret elaborated connections between values and computations, and Levy \cite{DBLP:phd/ethos/Levy01} used Freyd categories.
Although strong monad, Freyd category, and copower have a strong connection with \emph{monoidal action (actegory)}, actegory itself had not been treated as a ``first-class citizen'' as a model of effects.
Interestingly, because the expressive power is restricted enough, SEC can be directly modeled by equivariant functors (morphisms between actegories).
While only the soundness holds in our semantics (i.e.\ the completeness result presumably fails due to its ``lax'' nature), this semantics allows us to perform the term model construction.

Categorical gluing (also called sconing or Freyd cover in the literature) is a method to create a structure satisfying a certain categorical notion from a morphism preserving structures to which the notion is related.
Although categorical gluing is not always possible for every type of categorical structure, it is well known that cartesian closed structure (the structure of simply-typed lambda calculus) admits such gluing construction, with which one can prove some syntactic properties (e.g.\ conservativity) of the lambda calculus \cite{Crole, SconingAndRelators}.
Moreover, categorical gluing is closely related to logical predicates and logical relations \cite{SconingAndRelators, DBLP:conf/tlca/Hasegawa99}.
As a corollary of the present work, we prove that gluing along equivariant functors (\emph{actegorical gluing}) is indeed possible, and derive logical predicates for SEC.
We also show that in some typical cases, actegorical gluing can derive the same logical predicates constructed by $\top\top$-lifting for the metalanguage \cite{DBLP:conf/tlca/LindleyS05, DBLP:conf/csl/Katsumata05}.

Our contributions summarize as follows:
\begin{itemize}
  \item We present a proof-theoretic reformulation of intuitionistic possibility modality.
  \item We propose a new calculus (SEC) capturing a weaker notion of computation (semi-effects).
  \item We show that SEC is modeled by actegories and equivariant functors.
  \item We prove that categorical gluing along equivariant functors is possible, and show that the $\top\top$-lifting of strong monads is reducible to this gluing in some typical cases.
\end{itemize}

\subsection*{Construction of the paper}

Section 2 presents related work.
In Section 3, we recall Moggi's metalanguage and introduce our calculus SEC.
Their syntactic definitions and logical properties are discussed.
In Section 4, we present some basics of actegories and equivariant functors, and give categorical models for SEC.
Section 5 describes categorical gluing along models of SEC and its connection with fibrations.
As an application, we present certain flavors of logical predicates for SEC.
Comparison with $\top\top$-lifting is also presented here.
Section 6 concludes the paper and discusses future work.

\section{Related work}



\textbf{Proof-theoretic reconstruction of modal logic.}
It has been a long-standing issue to find a nice proof-theoretic account of intuitionistic modal logics.
Pfenning and Davies proposed a reformulation of intuitionistic modalities of both necessity $\Box$ and possibility $\Diamond$ \cite{DBLP:journals/mscs/PfenningD01}.
In the presence of both two modalities, a clean categorical account by an $\mathcal{L}$-strong monad is possible \cite{DBLP:journals/tcs/Kobayashi97}.
For the necessity-only fragment, a more refined calculus based on stratification of the modality and its semantics based on iterated enrichment of categories are presented in \cite{DBLP:journals/entcs/NishiwakiKM18,DBLP:journals/jip/KakutaniMN19}.
A brief survey of this field is found in \cite{DBLP:journals/corr/Kavvos16c}.

\noindent
\textbf{Categorical gluing.}
Mitchell and Scedrov \cite{SconingAndRelators} pointed out that the classical fundamental lemma for logical predicates is obtained as the uniqueness of the morphism from the classifying category of simply-typed $\lambda$-calculus to a cartesian closed category constructed by gluing (or sconing).
In \cite{Crole}, a conservativity proof of $\lambda\times$-calculus over equational logic of algebras is presented.
The proof uses the gluing technique and exploits the universal property of na\"ive translation from an algebraic theory to a $\lambda\times$-theory.
In \cite{DBLP:conf/ppdp/Fiore02} and \cite{DBLP:journals/corr/abs-1809-08646}, a normalization proof of simply-typed $\lambda$-calculus by gluing is presented.
Whether these techniques can be adapted to our results remains unclear.

\noindent
\textbf{Logical predicates.}
Logical predicates (and logical relations) have been used to prove syntactic results for many calculi (e.g.\ \cite{DBLP:journals/iandc/Statman85}).
Examples include the computational adequacy result of PCF \cite{Streicher}.
Hermida \cite{Hermida} generalized the logical predicates for simply-typed $\lambda$-calculus to cartesian closed categories using the internal logic and (Grothendieck) fibrations.
We extensively use the results from this work.
In \cite{DBLP:conf/tlca/Hasegawa99}, Hasegawa showed that logical predicates for certain fragments of linear logic can be described by \emph{subgluing} via fibrational arguments, from which some semantic results (i.e.\ the ability to obtain new models) for gluing and subgluing are derived.
Our present work is similar to Hasegawa's work.
There are some categorical formulations of logical predicate for monadic computation.
Our construction is closely related to \cite{DBLP:conf/csl/Katsumata05} (see Subsection \ref{sec:actegorical gluing and TT-lifting}).
On the other hand, the relationship between \cite{DBLP:journals/mscs/Goubault-LarrecqLN08} and the present work remains unknown.

\section{Semi-effect calculus}

In this section, we introduce \emph{Semi-Effect Calculus (SEC)}, which will be studied throughout the paper.
SEC is obtained by careful analysis of Moggi's metalanguage.

\subsection{Preliminaries on logical harmony and stability}

We recall some basic notions from proof-theoretic semantics (PTS).
The materials in this subsection will be necessary to understand the construction in the next subsection.

PTS is an approach to investigate the meaning of a logical constant (connective) by means of the structural nature of the natural deduction system associated to the logic.
Unlike traditional Tarski-style semantics, PTS is considered a rather informal, philosophically-motivated semantics.
Nonetheless, PTS is supposed to help more conceptual understanding of logics and provide a criterion for designing a well-behaved natural deduction system and hence the corresponding term calculus.

\emph{Logical harmony} (a la Dummett) in PTS is such a property that (it is expected that) every ``meaningful'' logical connective shall enjoy.
We consider Prior's \emph{tonk} \cite{Prior60}, which is an imaginary logical connective having the introduction rule (I-rule) of disjunction and the elimination rule (E-rule) of conjunction.
\[
  \AXC{$\Gamma \vdash A_i$}
  \RL{I-$\mathrm{tonk}_i$}
  \UIC{$\Gamma \vdash A_1 \mathrel{\mathrm{tonk}} A_2$}
  \DP
  \qquad
  \AXC{$\Gamma \vdash A_1 \mathrel{\mathrm{tonk}} A_2$}
  \RL{E-$\mathrm{tonk}_i$}
  \UIC{$\Gamma \vdash A_i$}
  \DP
\]
Clearly, having $\mathrm{tonk}$ makes the logic syntactically inconsistent (i.e.\ proves everything).
Some criteria have been proposed to answer why $\mathrm{tonk}$ is nonsense (and others are not).
Prawitz' \emph{inversion principle} (e.g.\ \cite{Read15}), a (candidate of) formulation of logical harmony, claims that an E-rule should not be ``weaker'' than the I-rule, in the sense that using the E-rule immediately after the I-rule should only prove propositions that are already in the premises of the I-rule.
$\mathrm{tonk}$ does not satisfy this property and hence is rejected.
There is also a converse criterion, called \emph{stability} \cite{Dummett91}, which states that an E-rule should not be ``too strong'' compared to the I-rule.
All meaningful connectives (including connectives in ordinary intuitionistic logic) are considered to enjoy both the inversion principle and stability.

\subsection{Lax logic}

Lax logic is an intuitionistic modal logic with one possibility modality operator $\dia$.
It features the following rules for $\dia$ along with the usual rules for intuitionistic propositional logic.
\[
    \AXC{$\Gamma \vdash_\textrm{lax} A$}
    \RL{I-$\dia$}
    \UIC{$\Gamma \vdash_\textrm{lax} \dia A$}
    \DP
    \qquad
    \AXC{$\Gamma \vdash_\textrm{lax} \dia A$}
    \AXC{$\Gamma, A \vdash_\textrm{lax} \dia B$}
    \RL{E-$\dia$}
    \BIC{$\Gamma \vdash_\textrm{lax} \dia B$}
    \DP
\]
Lax logic's significance is the Curry-Howard correspondence with Moggi's metalanguage.
Term assignment to I-$\dia$ and E-$\dia$ yields terms $\return{M}$ and $\letin{x}{M_1}{M_2}$ in the metalanguage in an evident way.
It is also shown that operational aspects of the metalanguage are easily adapted to well-known proof-theoretic notions (e.g.\ proof normalization) \cite{DBLP:journals/jfp/BentonBP98, DBLP:journals/iandc/FairtloughM97}.

According to PTS, however, this formalism of lax logic is unsatisfactory, for that $\dia$ is unstable.
Because I-$\dia$ proves $\dia A$ from any $A$, $\dia A$ is considered to have precisely the same information as $A$.
To be as strong as I-$\dia$, E-$\dia$ then must be such a rule that directly extracts $A$ from any $\dia A$, or dually, turns any sequent $\Gamma, A \vdash_{\mathrm{lax}} B$ with premise $A$ into $\Gamma \vdash_{\mathrm{lax}} B$ given $\Gamma \vdash_{\mathrm{lax}} \dia A$.
Clearly, the actual E-$\dia$ rule has an extra restriction on the form of conclusion, namely $\dia B$, by which stability fails.
(See \cite{Read15} for more details.)

\subsection{Simple adjoint calculus} \label{sec:sac}

Adjoint calculus \cite{DBLP:conf/csl/Benton94} is a calculus for linear logic that incorporates two styles of judgments, one for linear reasoning and the other for non-linear (classical) reasoning.
Exploiting the idea of adjoint calculus, we decompose the modality $\dia$ into a composite of two modalities $\rightadj \comp \leftadj$.
To this end, we restrict our focus to a fragment of lax logic where every judgment has precisely one premise.
This fragment Curry-Howard-corresponds to what is called the \emph{simple metalanguage} in Moggi's original paper \cite{Moggi91}, where every term has precisely one free variable.
In fact, the decomposition presented in the sequel is the same as restriction of adjoint calculus to the single variable fragments.
After this fact, we call the decomposed calculus \emph{simple adjoint calculus (SAC)}.

\begin{figure}
  \begin{typingrules}
    \AXC{$\type{x}{\tau}\vdashv\type{M}{\tau_1}$}
    \RL{$\morphism{f}{\tau_1}{\tau_2}$}
    \UIC{$\type{x}{\tau}\vdashv \type{f(M)}{\tau_2}$}
    \DP
    \quad
    \AXC{$\context{\Gamma}{\Delta}\vdashc\type{N}{A_1}$}
    \RL{$\morphism{g}{A_1}{A_2}$}
    \UIC{$\context{\Gamma}{\Delta}\vdashc \type{g(N)}{A_2}$}
    \DP
    \quad
    \AXC{$\type{x}{\tau}\vdashv\type{M}{\tau'}$}
    \RL{$\morphism{h}{\tau'}{A}$}
    \UIC{$\context{\type{x}{\tau}}{}\vdashc \type{h(M)}{A}$}
    \DP
    \quad
    \AXC{$\mathstrut$}
    \UIC{$\type{x}{\tau}\vdashv\type{x}{\tau}$}
    \DP
    \quad
    \AXC{$\mathstrut$}
    \UIC{$\context{{\cdot}}{\type{v}{A}}\vdashc\type{v}{A}$}
    \DP
    \quad
    \AXC{$\context{\type{x}{\tau}}{{\cdot}}\vdashc\type{N}{A}$}
    \UIC{$\type{x}{\tau}\vdashv\type{\reify{N}}{\rightadj A}$}
    \DP
    \quad
    \AXC{$\type{x}{\tau}\vdashv\type{M}{\rightadj A}$}
    \UIC{$\context{\type{x}{\tau}}{{\cdot}}\vdashc\type{\reflect{M}}{A}$}
    \DP
    \quad
    \AXC{$\type{x}{\tau}\vdashv\type{M}{\tau'}$}
    \UIC{$\context{\type{x}{\tau}}{{\cdot}}\vdashc \type{\val{M}}{\leftadj\tau'}$}
    \DP
    \quad
    \AXC{$\context{\Gamma}{\Delta}\vdashc\type{N_1}{\leftadj\tau}$}
    \AXC{$\context{\type{x}{\tau}}{{\cdot}}\vdashc\type{N_2}{A}$}
    \BIC{$\context{\Gamma}{\Delta}\vdashc \type{\letvalin{x}{N_1}{N_2}}{A}$}
    \DP
  \end{typingrules}
  \unskip
  \caption{Typing rules of SAC}
  \label{sac:typing rules}
  \centering
\end{figure}

Figure \ref{sac:typing rules} presents the complete list of typing rules of SAC.
In the figure, \fbox{$\context{\Gamma}{\Delta}$} denotes either \fbox{$\context{\type{x}{\tau}}{}$} or \fbox{$\context{}{\type{v}{A}}$}.
Therefore, a judgment in SAC is in one of the following forms of \fbox{$\type{x}{\tau}\vdashv \type{M}{\tau'}$}, \fbox{$\context{\type{x}{\tau}}{\emptycontext}\vdashc \type{N}{A}$}, or \fbox{$\context{\emptycontext}{\type{v}{A}}\vdashc \type{N}{A'}$}.
Note that every judgment has exactly one free variable.

To convey the intuition, we start by explaining the semantics first rather than syntactic details.
SAC's semantics is simply given by any adjunction between any categories:
\begin{center}
  \begin{tikzcd}
    \C \ar[r, bend left, "\leftadj"] \ar[r, phantom, "\rotatebox{90}{$\vdash$}"] & \D \ar[l, bend left, "\rightadj"]
  \end{tikzcd}.
\end{center}
As the symbols suggest, we identify the type operators $\leftadj$ and $\rightadj$ with the left and right adjoint functors in the model.
Namely, we identify a judgment $\vdashv$ with a morphism in $\C$ and $\vdashc$ with $\D$.
By identifying context \fbox{$\context{\type{x}{\tau}}{}$} with \fbox{$\context{}{\type{v}{\leftadj \tau}}$}, one may think of $\mathtt{val}$ as the functor $\leftadj$'s action on morphisms $\leftadj_{\tau,\tau'} \colon \C(\tau,\tau') \to \D(\leftadj\tau,\leftadj\tau')$.
Similarly, $\mathtt{reify}$ and $\mathtt{reflect}$ are identified with functions sending a morphism to its transpose.
While we have not yet introduced enough syntactic notions, the intention of the following statement should now be clear.
That is, SAC serves as an internal language of adjunctions.

\begin{theorem}
  There is a sound and complete interpretation of SAC in an adjunction.
\end{theorem}

Let us explain the syntax of SAC in detail.
All types in SAC are classified into two classes, which we call \emph{value} types and \emph{computation} types.
Note that these terminologies are arbitrary.
The model may no longer be a Kleisli adjunction, hance with no flavor of computation.
We call terms of value type (resp.\,computation type) \emph{value terms} (resp.\,\emph{computation terms}).
We use metavariables $M,M',\cdots$ for value terms, $N,N',\cdots$ for computation terms, and $L,L',\cdots$ for any terms.
A \emph{signature of SAC} consists of a set $V$ of base value types, a set $C$ of base computation types, and a set $F$ of function symbols.
Fixing a signature defines the sets of all value and computation types, which are freely generated by the base types and type operators $\leftadj$ and $\rightadj$.

Equations are given to typed terms as in the metalanguage.
We only consider equations between terms with the same type under the same context (i.e.\ equations-in-context).
The definitional equality (postulated equations) of SAC is given by the following rules.
The rules for congruence, reflexivity, symmetry, transitivity, and substitution are omitted for brevity.
\begin{align*}
  \reflect{(\reify{N})} &=_{A} N \tag{$\beta_\rightadj$} \\
  M &=_{\rightadj A} \reify{(\reflect{M})} \tag{$\eta_\rightadj$} \\
  (\letvalin{x}{\val{M}}{N}) &=_{A} N[M/x] \tag{$\beta_\leftadj$} \\
  N &=_{\leftadj\tau} (\letvalin{x}{N}{\val{x}}) \tag{$\eta_\leftadj$} \\
  (\letvalin{x_1}{C[N_1]}{N_2}) &=_{A} C[\letvalin{x_1}{N_1}{N_2}] \tag{comm.\,conv.}
\end{align*}
Here, $C[-]$ denotes any typed context.
Given a signature, a \emph{theory of SAC} is a set of equations-in-context in the signature.

Note that we can easily redefine $\mathtt{return}$ and $\mathtt{let}$ in the simple metalanguage in SAC:
\begin{align*}
  \return{M} &\defeq \reify{(\val{M})} \\
  (\letin{x}{M_1}{M_2}) &\defeq \reify{(\letvalin{x}{\reflect{M_1}}{M_2})}
\end{align*}
However, the converse is not possible for that the class of models is widened from any monads to any adjunctions.

Unlike lax logic, the modalities in SAC are considered stable.
The strength of the introduction and elimination rules of $\leftadj$ is equalized (at least in its succedents) in the sense that $\mathtt{val}$ creates $\leftadj \tau$ from any $\tau$ and $\mathtt{let\;val}$ destructs $\leftadj \tau$ into any proof term with a hole of type $\tau$.
This is also confirmed by checking the associativity rule of the metalanguage is rephrased with a commutative conversion rule with typed context $C[-]$ in SAC.

\subsection{Semi-effect calculus} \label{sec:sec}

Now that we have accomplished our proof-theoretic reconstruction, we further derive another calculus that is interesting as a programming language.
In SAC, we could freely switch back and forth between the realms of values (terms under $\vdashv$) and computations (terms under $\vdashc$).
By removing the rules for $\mathtt{reify}$ and $\mathtt{reflect}$ from the calculus and allowing multiple variables in the value context, we obtain a new calculus, which we dub \emph{semi-effect calculus (SEC)}.
In SEC, the realms of values and computations are no longer treated dually.
Instead values can only ``act'' on computations in a way we later justify via semantic arguments.
Still, the calculus has a flavor of computation as it incorporates $\mathtt{val}$ and $\mathtt{let\;val}$.
We call this phenomenon \emph{semi-effectful}.

As in SAC, the set of types in SEC is given by a set of value types and a set of computation types, denoted by $\tau$ and $A$ respectively:
\begin{align*}
  \tau &\bnfdefeq \sigma \mid \tau\times\tau\mid 1\\
  A    &\bnfdefeq b \mid \leftadj \tau
\end{align*}
where $\sigma$ is any base value type and $b$ is any base computation type.
Notice that we also assume finite product types in values.
Since the right adjoint modality $\rightadj$ is dropped, nested computation types such as $\rightadj\leftadj\rightadj\leftadj\tau$ are no longer valid.
Each function symbol has one of three sorts: $\morphism{}{\vec{\tau_i}}{\tau}$, $\morphism{}{\vec{\tau_i}, A}{A'}$, and $\morphism{}{\vec{\tau_i}}{A}$, where $\vec{\tau_i}$ denotes $\tau_1,\ldots,\tau_n$ for some $n \in \N$.

\begin{figure}
  \begin{typingrules}
    \AXC{$\Gamma \vdashv \type{M_1}{\tau_1}$}
    \AXC{$\cdots$}
    \AXC{$\Gamma \vdashv \type{M_n}{\tau_n}$}
    \RL{$\morphism{f}{\vec{\tau_i}}{\tau}$}
    \TIC{$\Gamma \vdashv \type{f(M_1,\ldots,M_n)}{\tau}$}
    \DP
    \quad
    \AXC{$\Gamma \vdashv \type{M_1}{\tau_1}$}
    \AXC{$\cdots$}
    \AXC{$\Gamma \vdashv \type{M_n}{\tau_n}$}
    \AXC{$\context{\Gamma}{\Delta} \vdashc \type{N}{A}$}
    \RL{$\morphism{g}{\vec\tau_i,A}{A'}$}
    \QIC{$\context{\Gamma}{\Delta} \vdashc \type{g(M_1,\ldots,M_n,N)}{A'}$}
    \DP
    \quad
    \AXC{$\Gamma \vdashv \type{M_1}{\tau_1}$}
    \AXC{$\cdots$}
    \AXC{$\Gamma \vdashv \type{M_n}{\tau_n}$}
    \RL{$\morphism{h}{\vec\tau_i}{A}$}
    \TIC{$\context{\Gamma}{} \vdashc \type{h(M_1,\ldots,M_n)}{A}$}
    \DP
    \quad
    \AXC{$\mathstrut$}
    \RL{$(\type{x}{\tau}) \in \Gamma$}
    \UIC{$\Gamma \vdashv \type{x}{\tau}$}
    \DP
    \quad
    \AXC{$\mathstrut$}
    \UIC{$\context{\Gamma}{\type{v}{A}} \vdashc \type{v}{A}$}
    \DP
    \quad
    \AXC{$\Gamma \vdashv \type{M}{\tau}$}
    \UIC{$\context{\Gamma}{{\cdot}} \vdashc \type{\val{M}}{\leftadj\tau}$}
    \DP
    \quad
    \AXC{$\context{\Gamma}{\Delta} \vdashc \type{N_1}{\leftadj\tau}$}
    \AXC{$\context{\type{x}{\tau},\Gamma}{{\cdot}} \vdashc \type{N_2}{A}$}
    \BIC{$\context{\Gamma}{\Delta} \vdashc \type{\letvalin{x}{N_1}{N_2}}{A}$}
    \DP
  \end{typingrules}
  \unskip
  \caption{Typing rules of \sec{} (rules for finite product types are omitted)}
  \label{sec:typing rules}
  \centering
\end{figure}

\figurename{} \ref{sec:typing rules} lists the typing rules of \sec{}.
A judgment in SEC has either of forms \fbox{$\Gamma\vdashv \type{M}{\tau'}$} or \fbox{$\context{\Gamma}{\Delta}\vdashc \type{N}{A}$}.
Here, $\Gamma$ is a context of zero or more value variables and $\Delta$ is a context of zero or one computation variable.
While contexts in SEC have unusual forms, the usual properties of typing judgment hold without difficulty.

\begin{lemma}
  The uniqueness of typing holds for both $\vdashv$ and $\vdashc$.
  The weakening, contraction, and exchange rules hold for the value context.
  The structural rule of substitution holds for both the value and computation contexts.
\end{lemma}

SEC inherits the equation-in-context rules from SAC.
Using concepts up to here, we can introduce the theory of \sec.

\begin{definition}
  A \emph{signature} of \sec{} is given by sets $V$ and $C$ of value/computation base types and a set $F$ of function symbols.
  A \emph{theory} of \sec{} consists of a signature $\Sigma$ and a set $\Ax$ of \emph{axioms}, well-formed equations under $\Sigma$.
\end{definition}

By seeing $\mathtt{let\,val}$ as $\mathtt{let}$ and $\mathtt{val}$ as $\mathtt{return}$, we can easily transport examples of Moggi's metalanguage (e.g.\ stateful, nondeterministic, and so on) to SEC.
Moreover, SEC can express a term that is not ``effectful'' but ``semi-effectful''.
Here we demonstrate this by showing an example using Haskell's \Applicative{} \cite{DBLP:journals/jfp/McbrideP08}, which is a generalization of \texttt{Monad}.

Recall that a functor \texttt{f} in Haskell is \Applicative{} if it is endowed with two operators
\begin{align*}
  \mathtt{pure} &\colon \mathtt{a} \to \mathtt{f\; a} \\
  \texttt{<*>} &\colon \mathtt{f\; (a \to b)} \to \mathtt{f\; a} \to \mathtt{f\; b}
\end{align*}
satysfing some laws.
A leading example that is not \texttt{Monad} but \Applicative{} is \texttt{ZipList}.
\texttt{ZipList a} is a type of finite or infinite sequence of type \texttt{a}.
Its associated \texttt{pure} is given by $\texttt{pure}\,x \defeq (x)_{i < \infty}$ and \texttt{<*>} is given by $(f_i)_{i < n} \texttt{<*>} (x_i)_{i < m} \defeq (f_i x_i)_{i < \min\{n,m\}}$.
Because \texttt{ZipList} is not \texttt{Monad}, we cannot use Moggi's metalanguage to reason about it.
On the other hand, in SEC, such reasoning is possible.
We define a theory for \texttt{ZipList} $\T_{\texttt{ZipList}}$ as the internal language of the Freyd category associated to the lax monoidal functor of \texttt{ZipList}, where we defer the technical details to Example \ref{ex:model of sec} and Corollary \ref{cor:internal language}.
Here we only point out that $\vdashc$ corresponds to the applicative context, whereas $\vdashv$ is the pure context.
Inside this theory, terms of \texttt{ZipList} can be defined in a style very much like applicative-do \cite{DBLP:conf/haskell/MarlowJKM16}:
\[
  \letvalin{x}{\mathtt{[1,2,3]}}{\letvalin{y}{\mathtt{[4,5]}}{\val{(x + y)}}}.
\]
This term roughly corresponds to the following expression in applicative-do:
\begin{equation}
  \label{eq:applicative do}
  \texttt{do \{ x <- [1,2,3]; y <- [4,5]; pure (x + y) \}}
\end{equation}
which is desugared to \texttt{pure (\textbackslash x y -> x + y) <*> [1,2,3] <*> [4,5]} and results in \texttt{[5,7]}.
For the sake of soundness, applicative-do disallows a term at the position of \texttt{[4,5]} in \eqref{eq:applicative do} to use \texttt{x}.
However, there is no such limitation in SEC, and thus the following is perfectly valid:
\[
  \letvalin{x}{\mathtt{[1,2,3]}}{\letvalin{y}{\val{(x + 1)}}{\val{(x + y)}}}.
\]
In this way, we obtain a logic of \texttt{ZipList} for free, in which we can reason e.g.\ as follows:
\begin{gather*}
  \begin{split}
    \context{\Gamma}{\Delta} \vdashc &\left(\letvalin{x}{\mathtt{[1,2,3]}}{\letvalin{y}{\mathtt{[4,5]}}{\val{(x + y)}}}\right) \\ &= \left(\letvalin{x}{\mathtt{[4,5]}}{\letvalin{y}{\mathtt{[1,2,3]}}{\val{(x + y)}}}\right).  
  \end{split}    
\end{gather*}

Note that SEC admits more models beyond \Applicative{}, as we will see in Section \ref{sec:model}.

\section{Categorical models for SEC}
\label{sec:model}

In this section, we introduce a categorical semantics of \sec.
Our semantics is built upon monoidal actions.
We fix a monoidal category $(\M,\tensor, I, r, l, a)$.

\begin{definition}[monoidal action, actegory, e.g.\ \cite{janelidze2001note}]
  Let $\C$ be a category.
  A bifunctor $\morphism{(-)\cdot (-)}{\M\times \C}{\C}$ is an \emph{$\M$-action} on $\C$ if there are natural isomorphisms $\morphism{\eta_c}{I\cdot c}{c}$ and $\morphism{\mu_{m_1,m_2,c}}{(m_1\tensor m_2)\cdot c}{m_1\cdot (m_2\cdot c)}$ making the following diagrams commute.
  \[
    \begin{tikzcd}[column sep=1.5cm]
      (m_1 \tensor m_2 \tensor m_3) \cdot c \ar[r, "\mu_{m_1 \tensor m_2, m_3, c}"] \ar[d, "\mu_{m_1, m_2 \tensor m_3, c}"'] & (m_1 \tensor m_2) \cdot (m_3\cdot c) \ar[d, "\mu_{m_1, m_2, m_3 \cdot c}"] \\
      m_1 \cdot ((m_2 \tensor m_3) \cdot c) \ar[r, "\id \cdot \mu_{m_2, m_3, c}"'] & m_1 \cdot (m_2 \cdot (m_3 \cdot c))
    \end{tikzcd}
    \quad
    \begin{tikzcd}[column sep=.2cm]
      (m \tensor I) \cdot c \ar[rr, "\mu_{m, I, c}"] \ar[rd, "r_m \cdot \id"'] && m \cdot (I \cdot c) \ar[ld, "\id \cdot \eta_c"] \\
      & m \cdot c
    \end{tikzcd}
  \]
  The left diagram implicitly uses the associativity $a$.
  An \emph{$\M$-actegory} is a category with a fixed $\M$-action on it.
\end{definition}

We often omit the prefix $\M$- from $\M$-action if it is inferrable from the context.

\begin{example}
  \label{ex:monoidal action}
  \begin{enumerate}
    \item 
      \label{ex:moncat has an action}
      Any monoidal category $\M$ is automatically an $\M$-actegory, where the action $\functor{(-)\cdot (-)}{\M\times\M}{\M}$ is given by the tensor product $\tensor$.
      
    \item
      Monoidal action subsumes the classical notion of monoid action.
      Any set is identifieid with a (small) discrete category and any monoid $(M, *, e)$ is identified with a monoidal category whose underlying category is discrete and whose tensor is given by $*$.
      Under this identification, a set $A$ is an $M$-actegory if and only if $A$ has a monoid action of $M$.
  \end{enumerate}
\end{example}

Morphisms of actegories are defined in the following sense.

\begin{definition}
  Let $\C,\D$ be $\M$-actegories.
  A functor $\morphism{F}{\C}{\D}$ is (lax) $\M$-equivariant (resp. strong $\M$-equivariant) if there is a coherent natural transformation (resp. isomorphism) $\morphism{\phi^F_{m, c}}{m\action F(c)}{F(m\action c)}$.
  We mean by \emph{coherence} that the diagrams below commute.
  \[
    \begin{tikzcd}
      (m \tensor m') \cdot Fc \ar[r, "\mu_{m, m', Fc}"] \ar[d, "\phi^F_{m \tensor m', c}"'] & m \cdot (m' \cdot Fc) \ar[r, "\id \cdot \phi^F_{m', c}"] & m \cdot F(m' \cdot c) \ar[d, "\phi^F_{m, m' \cdot c}"] \\
      F((m \tensor m') \cdot c) \ar[rr, "F\mu_{m, m', c}"'] && F(m \cdot (m' \cdot c))
    \end{tikzcd}
    \quad
    \begin{tikzcd}[column sep=.2cm]
      & I \cdot Fc \ar[dl, "\phi^F_{I,c}"'] \ar[dr, "\eta_{Fc}"] \\
      F(I \cdot c) \ar[rr, "F\eta_c"'] && Fc
    \end{tikzcd}
  \]
\end{definition}

We will omit the superscript $F$ for $\phi^F$ when this does not make confusion.
\emph{Strict $\M$-equivariant functor} is also defiend in the same mannar.

\begin{example}
  \label{ex:ef}
  \begin{enumerate}
    \item 
      Consider the canonical $\M$-action on $\M$ (see Example \ref{ex:monoidal action}).
      An equivariant functor $\functor{F}{\M}{\M}$ is precisely a \emph{strong} functor \cite{Kock1972StrongFA} $F$ on $\M$, where the strength $\functor{t_{A, B}}{A\tensor FB}{F(A\tensor B)}$ is $\phi^F_{A,B}$.

    \item
      Freyd category or value/producer structure \cite{DBLP:phd/ethos/Levy01} is a special case of strong equivariant functor.
      A Freyd category is an identity-on-objects functor $\functor{J}{\V}{\C}$ such that
      (1) $\V$ has finite products,
      (2) $\C$ has a $\V$-action, and
      (3) The $v\times (-)$ can be extended to the $\V$-action on $\C$ along $J$ for any $v\in \V$.
      These conditions say that $J$ is strict $\V$-equivariant.
  \end{enumerate}
\end{example}

Given a strong $\V$-equivariant functor $\morphism{\leftadj}{\V}{\C}$ where $\V$ has finite products and $\C$ has a $\V$-action w.r.t.\ the cartesian structure of $\V$, we can interpret theories
of \sec{}.
The interpretation follows the traditional category-of-contexts paradigm.
It is defined inductively once we fix data for base types and function symbols.
We will use $\V$ to interpret types and terms in the realm of values, and use $\C$ for the realm of computations.

Types and contexts in the realm of values are interpreted in $\V$ as usual: $\semantics{\tau_1 \times \tau_2} \defeq \semantics{\tau_1} \times \semantics{\tau_2}$ and $\semantics{\Gamma} \defeq \prod_{(\type{x_i}{\tau_i}) \in \Gamma} \semantics{\tau_i}$.
We use $\Gamma$ and $\tau$ almost interchangeably by this identification.
Computation types of the form $\leftadj \tau$ are interpreted using the functor $\leftadj$ by $\semantics{\leftadj\tau} \defeq \leftadj\semantics{\tau}$.
The two kinds of computation context have different interpretations:
$(\context{\Gamma}{\type{v}{A}})$ is interpreted by the action $\interpret{\context{\Gamma}{\type{v}{A}}} \defeq \interpret{\Gamma}\cdot \semantics{A}$, and $(\context{\Gamma}{\emptycontext})$ is interpreted by application of the equivariant functor $\interpret{\context{\Gamma}{\emptycontext}} \defeq \leftadj\interpret{\Gamma}$.

As to (well-typed) terms, we only show the case of computation terms (the case of value terms is rather obvious).
The interpretation of computation variables just discards value variables: $\semantics{\context{\Gamma}{\type{v}{A}}\vdashc \type{v}{A}} \defeq \interpret{\Gamma}\cdot\semantics{A}
\xrightarrow{\bang\cdot\id}1\cdot\semantics{A}\xrightarrow{\eta}\semantics{A}$.
$\mathtt{val}$ sends a value term to a computation term with the functor: $  \interpret{\context{\Gamma}{\emptycontext}\vdashc\type{\val{M}}{\leftadj\tau}} \defeq \leftadj\interpret{\Gamma\vdashv\type{M}{\tau}}$.
The most involved case is $\mathtt{let}\;\mathtt{val}$. $\interpret{\context{\Gamma}{\emptycontext}\vdashc\type{\letvalin{x}{N_1}{N_2}}{A}}$ is given by:
\[
  \leftadj\semantics{\Gamma} \xrightarrow{\leftadj\delta} \leftadj(\semantics{\Gamma} \times \semantics{\Gamma}) \xrightarrow{\phi^{-1}} \semantics{\Gamma} \cdot \leftadj\semantics{\Gamma} \xrightarrow{\id \cdot \semantics{N_1}} \semantics{\Gamma} \cdot \leftadj\semantics{\tau} \xrightarrow{\phi} \leftadj(\semantics{\Gamma \times \tau}) \xrightarrow{\semantics{N_2}} \semantics{A}.
\]
In the same vein, we can interpret the case when $\letvalin{x}{N_1}{N_2}$ has a free computation variable.

Given an interpretation $\semantics{-}$ of $\T$, an equation-in-context is defined to be \emph{valid w.r.t.\ $\semantics{-}$} if the two terms are externally equal i.e.\ interpreted by the same morphism.
It then follows that this interpretation is indeed sound.

\begin{theorem}
  \label{thm:model by ef}
  Let $\T$ be any theory of \sec{} and $\interpret{-}$ be an interpretation of $\T$.
  Assume that all axioms of $\T$ are satisfied by $\interpret{-}$.
  Then all equations derivable in $\T$ are satisfied by $\interpret{-}$.
\end{theorem}

\begin{proof}
  The proof is tedious but routine.
  One needs to be careful of whether $\Delta$ is empty or non-empty.
  (See Appendix for the detailed proof.)
\end{proof}

We will call such a strong $\V$-equivariant functor $\functor{\leftadj}{\V}{\C}$ a \emph{model of \sec{}}.

\begin{example} \label{ex:model of sec}
  \begin{enumerate}
    \item Given a strong monad on a cartesian category $\C$,
      we get a model of \sec{} by the Kleisli construction (see Example \ref{ex:ef}).
      By the definition of Kleisli category, a term \fbox{$\context{\Gamma}{\type{v}{\leftadj \tau}}\vdashc \type{N}{\leftadj\tau'}$}
      is interpreted by a morphism $\morphism{}{\interpret{\Gamma}\times \semantics{\tau}}{T\semantics{\tau'}}$ in $\C$.
      Furthermore, $\semantics{\val{M}}=\eta\comp\semantics{M}$ and $\semantics{\letvalin{x}{N_1}{N_2}}= \semantics{N_2}^{\#}\comp t\comp \pair{\id}{\semantics{N_1}}$ hold where $\semantics{N_2}^{\#}$ is the Kleisli lifting of $\semantics{N_2}$ and $N_2$ has no free computation variables. These interpretaions agree with those of $\semantics{\return{M}}$ and $\semantics{\letin{x}{N_1}{N_2}}$ in \cite{Moggi91}.

    \item
      There is a model of \sec{} that is not a Freyd category.
      The simplest is the inclusion $\iota_1 \colon 1 \to 1 + 1$ where $1$ is the terminal category and the action $\ast \cdot (-)$ is the identity.
    
    \item
      \label{ex:model from lax monoidal}
      It is folklore that a lax monoidal functor $F$ on a CCC $\C$ induces a Freyd category $J \colon \C \to \D$ \cite{DBLP:journals/entcs/HeunenJ06,DBLP:journals/entcs/LindleyWY11}.
      A morphism $f \colon X \to Y$ in $\D$ is given by a morphism $f \colon 1 \to F(Y^X)$ in $\C$.
      (A similar construction is also found in the semantics of multi-staged computation \cite{DBLP:journals/entcs/NishiwakiKM18}.)
      Because a lax monoidal functor models $\mathtt{Applicative}$ in Haskell, this serves as a model of the example presented at the end of Section 3.
      In this sense we consider SEC is semi-effectful, admitting more models than what were not supported by traditional models of effects, namely monads.
  \end{enumerate}
\end{example}

Every model of \sec{} gives rise to its \emph{internal language}, a theory of \sec{} such that all objects and morphisms of the model are base types and function symbols and $\Ax$ contains all such equations-in-context $L_1 = L_2$ that $\semantics{L_1} = \semantics{L_2}$.

\begin{corollary}
  \label{cor:internal language}
  Let $\T_F$ be the internal language of model $F$.
  The following are equivalent.
  \begin{itemize}
    \item $\T_F \vdash L_1 = L_2$ (i.e., equation-in-context $L_1 = L_2$ is derivable in $\T_F$)
    \item $\semantics{L_1} = \semantics{L_2}$ in $F$.
  \end{itemize}
\end{corollary}

Conversely, we can construct a strong equivariant functor from any theory of \sec.

\begin{theorem}
  \label{thm:theory to ef}
  Any theory of \sec{} induces a strong equivariant functor.
\end{theorem}

\begin{proof}
  We perform the term model construction as follows.
  The value category $\V$ is constructed as usual from value terms (see e.g.\ \cite{Crole}).
  The construction of the computation category $\C$ is somewhat tricky; it is defined by case distinction of computation context:
  \begin{align*}
    \C((\context{\Gamma}{A}),(\context{\tau}{A'}))
      &\defeq \{(\classof{\Gamma \vdashv \type{M}{\tau}}, \classof{\context{\Gamma}{A} \vdashc \type{N}{A'}})\} \\
    \C((\context{\Gamma}{\cdot}),(\context{\Gamma'}{\cdot}))
      &\defeq \V(\Gamma,\Gamma') \\
    \C((\context{\Gamma}{A}),(\context{\Gamma'}{\cdot}))
      &\defeq \emptyset
  \end{align*}
  where $\classof{\cdots}$ denotes the equivalence class of judgments up to the definitional equality.
  The $\V$-action on $\C$ is then given by $\tau \cdot (\context{\Gamma}{\Delta}) \defeq (\context{\tau, \Gamma}{\Delta})$.
  The equivariant functor $\leftadj$ sends $\tau$ to $(\context{\tau}{\cdot})$.
  See Appendix for the detailed construction.
\end{proof}

\begin{remark}
  \label{rem:failure of completeness}
  The crucial point of our term model construction is that the syntactic functor $\leftadj$ is defiend to be $\leftadj(\tau) \defeq (\context{\tau}{\emptycontext})$ instead of $(\context{\emptycontext}{\leftadj\tau})$.
  In fact, setting $\leftadj(\tau) \defeq (\context{\emptycontext}{\leftadj\tau})$ only gives us a lax equivariant functor.
  However, by setting $\leftadj(\tau) \defeq (\context{\tau}{\emptycontext})$, it in turn no longer holds that the term model interprets a term by itself.
  For example, $\semantics{\val M}$ is given by $M$ instead of $\val M$.
  As a result, Theorem \ref{thm:theory to ef} does not imply completeness of our semantics.
\end{remark}

Before proceeding, we introduce the notion of morphism of models of \sec{}.

\begin{definition}[equivariant natural transformation]
  \label{def:morphism of lef}
  Let $\C, \D$ be $\M$-actegories and $F, G: \functor{}{\C}{\D}$ be lax equivariant functors.
  A natural transformation $\morphism{\theta}{F}{G}$ is \emph{equivariant} if $\theta_{m\cdot c} \comp \phi^F_{m,c} = \phi^G_{m,c} \comp (m \cdot \theta_c)$ holds for all $m \in \M$ and $c \in \C$.
\end{definition}

$\M$-actegories, lax equivariant functors, and equivariant natural transformations form a 2-category $\Act{\M}$.
Replacing ``lax equivariant functors'' with ``strong equivariant functors'' yields another 2-category.
Any (strong) monoidal functor $\functor{F}{\M}{\M'}$ induces the change-of-base 2-functor $\functor{F^\ast}{\Act{\M'}}{\Act{\M}}$.

\begin{example}
  \label{ex:2-category of actegories}
  \begin{enumerate}
    \item
      For any 2-categorical notion X, we call an X internal to $\Act{\M}$ an \emph{$\M$-equivariant X}.
      A strong monad $T$ on $\M$ is just an $\M$-equivariant monad.
      Also, the Kleisli resolution $J \dashv K$ of $T$ is an $\M$-equivariant adjunction.
      Note that every equivariant left adjoint is strong equivariant.
      $J$ is a model of \sec{} in this way.

    \item
      Change of base along a strong monoidal functor $\functor{F}{\M}{\M'}$ makes $F$ a strong $\M$-equivariant functor, since $m \cdot_{\M'} F(m') = F(m) \tensor_{\M} F(m') \cong F(m \tensor_{\M} m') = F(m \cdot_{\M} m')$.
  \end{enumerate}
\end{example}

A morphism of models of \sec{} is defined in the language of actegories.

\begin{definition}[morphism of models of \sec]
  \label{def:morphism of models}
  Let $\V$ be a category with finite products and $\functor{F}{\V}{\C}, \functor{F'}{\V}{\C'}$ be models of \sec{}.
  A \emph{morphism of models of \sec{}} is given by a lax equivariant functor $\functor{H}{\C}{\C'}$ and an equivariant natural transformation $\morphism{\theta}{F'}{HF}$.
  \begin{center}
    \begin{tikzpicture}
      \matrix[row sep=.5cm, column sep=.7cm]{
        & \node (V) {$\V$};\\
        \node (C){$\C$}; & & \node (s) {$\C'$};\\
      };
      \coordinate (empty) at ($(V) !.6! (s)$);
      \draw 
        (V) edge node [above left] {$F$} (C)
            edge node [above right] {$F'$} (s)
        (C) edge node (a) [below, align=center] {$H$} (s);
  
      \node (arrow)
        at ($(C)!.7!(empty)$) [inner sep=0mm, rotate=30, label=$\theta$] {$\Leftarrow$};
    \end{tikzpicture}
  \end{center}
\end{definition}

Definition of 2-cells of models of \sec{} is omitted.
Note that every morphism of models is a morphism in $\int \Act{\V}^{\mathrm{co}}(\V,-)$, where $\int$ is the Grothendieck construction.

\begin{remark}
  \label{rem:failure of CHL}
  For reasons similar to Remark \ref{rem:failure of completeness}, the Curry-Howard-Lambek correspondence fails in our semantics.
  Specifically, we do not have a (canonical) equivalence of models of \sec{}: $F \xrightarrow{\simeq} \mathrm{Syn}(\T_F)$, where $\mathrm{Syn}(\T)$ is the term model of theory $\T$.
\end{remark}

\section{Categorical gluing for (lax) equivariant functors}

\subsection{Categorical gluing}

\emph{Categorical gluing} (also known as \emph{sconing}) is a technique to obtain a new model from a morphism of models.
It is a special case of the comma construction (see e.g.\ \cite{maclane:71}).

\begin{definition}[categorical gluing]
  \label{def:gluing}
  Given a functor $\functor{\Gamma}{\C}{\D}$, the \emph{gluing category of $\C$ to $\D$ along $\Gamma$} is obtained as a comma category $\glue{\D}{\Gamma}$.
  %
  The gluing category is equipped with the projection functor $\functor{\pi}{(\glue{\D}{\Gamma})}{\C}$.
\end{definition}

One important and interesting fact about gluing is that the gluing category $\glue{\D}{\Gamma}$ often inherits the involved structures.
In other words, when $\C$ and $\D$ have a certain kind of categorical structure and $\functor{\Gamma}{\C}{\D}$ preserves it, the gluing category $\glue{\D}{\Gamma}$ often has the structure and the projection functor $\pi$ preserves it.

\begin{example}
  \label{ex:gluing with fp}
  Let $\C$ and $\D$ be categories with (chosen) finite products and
  $\functor{\Gamma}{\C}{\D}$ be a functor preserving them (up to isormophism).
  Then the gluing category $\glue{\D}{\Gamma}$ has finite products and the projection functor $\pi$ (strictly) preserves them.
  Specifically, the terminal object is given by $(1, 1, \gamma)$, and the binary product of $(d, c, f)$ and $(d', c', f')$ is given by to by
  \[
     \bigl(d\times d', c\times c', d\times d'\xrightarrow{f\times f'}\Gamma c\times \Gamma c'\xrightarrow{\gamma_{c, c'}}\Gamma (c\times c')\bigr)
  \]
  where $\gamma$ and $\gamma_{c, c'}$ are the associated isomorphisms.
\end{example}

An important variance of gluing is \emph{subgluing} \cite{DBLP:conf/tlca/Hasegawa99}.

\begin{definition}[subgluing]
  \label{def:subgluing}
  Suppose a functor $\functor{\Gamma}{\C}{\D}$ is given. 
  By restricting the objects in $\glue{\D}{\Gamma}$ to subobjects, we get the full subcategory $\subglue{\D}{\Gamma}$ of $\glue{\D}{\Gamma}$.
  In other words, $(D, C, \morphism{f}{D}{\Gamma C})$ is an object in $\subglue{\D}{\Gamma}$ if and only if $f$ is a subobject.
  This category $\subglue{\D}{\Gamma}$ is called the subgluing of $\C$ to $\D$ along $\Gamma$.
\end{definition}

The gluing category and subgluing category for $\functor{\Gamma}{\C}{\D}$ are obtained as a pullback \cite{DBLP:conf/tlca/Hasegawa99}.

\begin{center}
  \begin{tikzpicture}
    \matrix[row sep=1cm, column sep=1cm]{
      \node (glue) {$(\glue{\D}{\Gamma})$}; & \node (arrow) {$\arrowcat{\D}$};\\
      \node (c) {$\C$}; & \node (d) {$\D$};\\
    };
    \draw
      (glue)  edge (arrow)
              edge (c)
      (c)     edge node [below] {$\Gamma$} (d)
      (arrow) edge node [right] {$\cod$} (d);
    \node at ($(glue) + (.4, -.5)$) {$\lrcorner$};
  \end{tikzpicture}
  \begin{tikzpicture}
    \matrix[row sep=1cm, column sep=1cm]{
      \node (glue) {$\subglue{\D}{\Gamma}$}; & \node (arrow) {$\Sub(\D)$};\\
      \node (c) {$\C$}; & \node (d) {$\D$};\\
    };
    \draw
      (glue)  edge (arrow)
              edge (c)
      (c)     edge node [below] {$\Gamma$} (d)
      (arrow) edge node [right] {$\Sub$} (d);
    \node at ($(glue) + (.4, -.5)$) {$\lrcorner$};
  \end{tikzpicture}
\end{center}

\subsection{Actegorical gluing} \label{sec:actegorical gluing}

We are able to present categorical gluing for actegories and lax equivariant functors.
First we show that gluing along lax equivariant functors yields an $\M$-actegory.

\begin{proposition}[actegorical gluing]
  \label{prop:simple gluing}
  Let $\C, \D$ be $\M$-actegories and $\functor{\Gamma}{\C}{\D}$ a lax equivariant functor.
  The gluing category $\glue{\D}{\Gamma}$ is an $\M$-actegories and the projection functor $\functor{\pi}{\glue{\D}{\Gamma}}{\C}$ is strict equivariant.
\end{proposition}

Proposition \ref{prop:simple gluing} can be generalized in terms of fibration.

\begin{proposition}
  \label{prop:lax equivariant and opfibration}
  Let $\B, \C, \E$ be $\M$-actegories and $\functor{\Gamma}{\C}{\B}$ a lax equivariant functor.
  In addition, let $\functor{p}{\E}{\B}$ be a strict equivariant functor which is also an opfibration, and suppose the condition $(*)$ holds.
  \begin{itemize}
    \item[$(*)$]
    For any object $m$ in $\M$, the functor $\functor{m\cdot (-)}{\E}{\E}$ preserves opcartesian morphisms.
  \end{itemize}
  Consider the pullback diagram in $\Cat$ below.
  \begin{center}
    \begin{tikzpicture}
      \matrix[row sep=.5cm, column sep=1cm]{
        \node (G) [label=below right:$\lrcorner$] {$\G$}; & \node (E) {$\E$};\\
        \node (C) {$\C$}; & \node (B) {$\B$};\\
      };
      \draw 
        (G) edge (E)
        edge node [left] {$q$} (C)
        (E) edge node [right] {$p$} (B)
        (C) edge node [below] {$\Gamma$} (B);
    \end{tikzpicture}
  \end{center}
  In this diagram, $\G$ has an $\M$-action and the functor $\functor{q}{\G}{\C}$ is strict equivariant.
\end{proposition}

\begin{proof}
  The $\M$-action on $\G$ is defined using the universality of opcartesian lifting of $\morphism{\phi^{\Gamma}}{m\cdot\Gamma C}{\Gamma (m\cdot C)}$.
  Notice $m\cdot \Gamma C = m\cdot pX=p(m\cdot X)$ for any $C\in \C$ and $X\in \E$ satisfying $\Gamma C=pX$.
  See Appendix for the detailed proof.
\end{proof}

Notice that Proposition \ref{prop:simple gluing} is just an instance of Proposition \ref{prop:lax equivariant and opfibration} when $p$ is the codomain opfibration $\functor{\cod}{\arrowcat{\D}}{\D}$.




\begin{remark}
  \label{rmk:lax equivariant and fibration}
  The ``opfibration'' in the statement of Proposition \ref{prop:lax equivariant and opfibration} cannot be simply replaced by ``fibration'' because the coherent natural transformation $\phi$ for $\Gamma$ is the form of $\morphism{\phi_{m, C}}{m\cdot \Gamma C}{\Gamma(m\cdot C)}$ and the cartesian lifting of $\phi_{m, C}$ cannot be considered in contrast to the opcartesian lifting.
\end{remark}

Although a pullback along a lax equivariant functor does not inherit the action as stated in Remark \ref{rmk:lax equivariant and fibration}, if we restrict $\Gamma$ to a strong equivariant one, we can get a similar proposition to Proposition \ref{prop:equivariant and fibration}.
Moreover, the condition $(*)$ can be dropped.

\begin{proposition}
  \label{prop:equivariant and fibration}
  Let $\B, \C, \E$ be $\M$-actegories and $\functor{\Gamma}{\C}{\B}$ a strong equivariant functor.
  In addition, let $\functor{p}{\E}{\B}$ be a fibration that is strict equivariant.
  Consider the diagram in Proposition \ref{prop:lax equivariant and opfibration}.
  In the diagram, $\G$ has an $\M$-action and the functor $\functor{q}{\G}{\C}$ is strict equivariant.
\end{proposition}

Note that, when $p$ is a bifibration in the situation of Proposition \ref{prop:equivariant and fibration}, there are two ways to define an $\M$-action on $\G$ by Proposition \ref{prop:lax equivariant and opfibration} and \ref{prop:equivariant and fibration}.
These coincide in the sense that they are isomorphic in $\Act{\M}$.
We can also consider subgluing.

\begin{proposition}
  \label{prop:subgluing for lef and ef}
  Consider the assumption of Proposition \ref{prop:simple gluing}.
  Assume moreover that functors $v\cdot (-)$ preserve monos for all $v \in \V$.
  If either of the following holds, $\subglue{\D}{\Gamma}$ has an $\M$-action and $\functor{\pi}{\subglue{\D}{\Gamma}}{\C}$ is strict equivariant.
  \begin{enumerate}
    \item \label{item:subgluing EM factor} $\D$ admits epi-mono factorization.
    \item \label{item:subgluing monic} $\phi^\Gamma$ is (componentwise) monic.
  \end{enumerate}
\end{proposition}

\begin{proof}
  Condition \ref{item:subgluing EM factor} makes the subobject functor $\functor{}{\sub{\D}}{\D}$ an opfibration, so $\subglue{\D}{\Gamma}$ inherits the action by the preceding propositions.
  Under Condition \ref{item:subgluing monic} an action is directly defined by extension with $\phi^\Gamma$.
\end{proof}

\begin{proposition}
  \label{prop:model by morphism of model}
  A morphism $(\functor{H}{\C}{\D}, \morphism{\theta}{F'}{HF})$ of models of \sec{} induces another model of \sec{}, i.e.\ a strong equivariant functor $\functor{L}{\V}{(\glue{\D}{H})}$ by the following construction.
  \begin{align*}
    Lv &\coloneqq (F'v, Fv, \morphism{\theta_{v}}{F'v}{HFv})\\
    L(\morphism{a}{v}{u}) &\coloneqq \morphism{(F'a, Fa)}{\theta_{v}}{\theta_{u}}
  \end{align*}
\end{proposition}

In fact, $L$ in Proposition \ref{prop:model by morphism of model} can also be described more conceptually from the point of view of the codomain fibration. 
Assume the situation in Proposition \ref{prop:model by morphism of model}.
We define a functor $\functor{G}{\V}{\arrowcat{\D}}$ by $Gv \defeq (\morphism{\theta_{v}}{F'v}{HFv})$.
This is well-defined by the naturality of $\theta$ and fanctoriality of $F$ and $F'$.
By definition of $G$, $\cod \comp G=HF$ holds where $\cod$ is the codomain functor $\functor{}{\arrowcat{\D}}{\D}$.
Since $\glue{\D}{H}$ is a pullback of $\cod$ along $H$, we obtain by universality a mediating functor $\functor{L}{\V}{(\glue{\D}{H})}$, which indeed coincides with $L$ in Proposition \ref{prop:model by morphism of model}.

\begin{center}
  \begin{tikzpicture}
    \matrix[row sep=.6cm, column sep=.8cm]{
      \node (v) {$\V$};\\
      & \node (g) [label=below right:$\lrcorner$] {$\glue{\D}{H}$}; & \node (arrow c') {$\arrowcat{\D}$};\\
      & \node (c) {$\C$}; & \node (c') {$\D$};\\
    };
    \draw 
      (v) edge [dashed] node [below, near start] {$L$} (g)
          edge [bend right=30] node [left] {$F$} (c)
          edge [bend left=10] node [below] {$G$} (arrow c')
      (g) edge node [right] {$\pi$} (c)
          edge node [below] {$\pi'$} (arrow c')
      (arrow c') edge node [left] {$\cod$} (c')
      (c) edge node [below] {$H$} (c');
    \draw [->]
      (v) .. node [above right, near end] {$F'$} controls +(2, 0) and +(2.5, 2.8)  .. (c');
    \node at (1.5, .7) [inner sep=0mm, rotate=30, label=$\theta$] {$\Leftarrow$};
  \end{tikzpicture}
\end{center}

Combining Proposition \ref{prop:subgluing for lef and ef} and \ref{prop:model by morphism of model}, we obtain the subgluing version of Proposition \ref{prop:model by morphism of model}.

\begin{proposition}
  \label{prop:model by morphism of model for subgluing}
  Let $(\functor{H}{\C}{\D}, \morphism{\theta}{F'}{HF})$ be a morphism of models of \sec{} such that $\theta$ and $\phi^H$ are (componentwise) monic and the functors $v\cdot (-)$ preserve monos for all $v \in \V$.
  Then we obtain a $\V$-actegory $\subglue{\D}{H}$ by Proposition \ref{prop:subgluing for lef and ef}.
  Moreover, we can get a model of \sec{} i.e.\ a strong equivariant functor $\functor{L}{\V}{\subglue{\D}{H}}$.
\end{proposition}

\begin{proof}
  We use Condition \ref{item:subgluing monic} from Proposition \ref{prop:subgluing for lef and ef}.
\end{proof}

\subsection{Flavor of logical predicates}

As a toy example of the tools developed up to here, we derive logical predicates for \sec{}.
We fix a theory $\T$ of \sec.

Let $\functor{F}{\V}{\C}$ be the term model of $\T$ (see Theorem \ref{thm:theory to ef}).
For some object $(\context{\Gamma}{\emptycontext})$ in $\C$, consider the functor $\functor{H=\C((\context{\Gamma}{\emptycontext}), -)}{\C}{\Sets}$.
By the definition of the term model, this functor can be made a morphism of models as follows.
To make $\Sets$ a $\V$-actegory, we define its action by the functor $\functor{\V(\Gamma, -)\times\id_{\Sets}}{\V\times\Sets}{\Sets}$.
This bifunctor is indeed a $\V$-action on $\Sets$ by the universality of products.
A model of \sec{} over $\Sets$ (i.e.\ a strong equivariant functor $\functor{}{\V}{\Sets}$) is given by $\V(\Gamma, -)$.
Again, it is easy to see that this is indeed a model.
Finally, by giving $\theta_{\tau}(\classof{\Gamma\vdashv\type{M}{\tau}}) = \classof{\Gamma\vdashv\type{M}{\tau}}$, $H$ becomes a morphism of models from $F$ to $\V(\Gamma, -)$.
In this situation, we obtain the subgluing model $\functor{L}{\V}{\subglue{\Sets}{H}}$ by Proposition \ref{prop:model by morphism of model for subgluing}.
Let us give an explicit definition of $L$.
For an object $\tau$ in $\V$, $L(\tau)$ is given by:
\[
  L(\tau) \defeq \V(\Gamma, \tau) \subseteq \C((\context{\Gamma}{\emptycontext}), (\context{\tau}{\emptycontext})).
\]
To see what the $\V$-action on $\subglue{\Sets}{H}$ does, take any objects $\tau \in \V$ and $P \in \subglue{\Sets}{H}$ where $P \subseteq \C((\context{\Gamma}{\emptycontext}),(\context{\tau'}{\Delta}))$ for some $(\context{\tau'}{\Delta}) \in \C$.
When $\Delta$ is not empty, say $\Delta=A$, $\tau \cdot P$ is given by
\[
  \biggl\{
    \left(\classof{\Gamma \vdashv \type{\pair{M}{M'}}{\tau \times \tau'}}, \classof{\context{\Gamma}{\emptycontext} \vdashc \type{N}{A}}\right)
    \biggmid
    (\classof{\Gamma\vdashv\type{M'}{\tau'}}, \classof{\context{\Gamma}{\emptycontext}\vdashc\type{N}{A}}) \in P
  \biggr\}\raisebox{-2ex}{.}
\]
When $\Delta$ is empty, $\tau \cdot P$ is defined in the same way by using only the first component.

To summarize, we obtain the definition of logical predicates for \sec{} as follows.
\begin{definition}
  Let $\Gamma$ be a value context.
  We define a set $I$ by
  \[
    I = \ob\C_{\T} \cup \{A\mid \text{$A$ is a computation type}\}.
  \]
  An $I$-indexed family of predicates $\{P_{i}\}_{i\in I}$ is a \emph{logical predicate for $\T$} if:
  \begin{itemize}
    \item 
      $P_b \subseteq \biggl\{
        \classof{\context{\Gamma}{\emptycontext} \vdashc \type{N}{b}}
      \biggr\}$

    \item
      $P_{\leftadj\tau}=P_{(\context{\tau}{\emptycontext})}=\biggl\{
        \classof{\Gamma \vdashv \type{M}{\tau}}
      \biggr\}$

    \item 
      $P_{(\context{\tau}{b})} = \biggl\{
        \left(\classof{\Gamma \vdashv \type{M}{\tau}}, \classof{\context{\Gamma}{\emptycontext} \vdashc \type{N}{A}}\right)
        \biggmid
        \classof{\context{\Gamma}{\emptycontext}\vdashc\type{N}{b}} \in P_b
      \biggr\}$

    \item 
      $P_{(\context{\tau}{\leftadj\tau'})} = \biggl\{
        \left(\classof{\Gamma \vdashv \type{M}{\tau}}, \classof{\Gamma \vdashv \type{M'}{\tau'}}\right)
        \biggmid
        \classof{\Gamma \vdashv \type{M'}{\tau'}} \in P_{\tau'}
      \biggr\}$

    \item
      A function $\fun{\interpret{g}}{P_{(\context{\prod\tau_i}{\Delta})}}{P_{A}}$
      for each function symbol $\morphism{g}{\vec\tau_i, \Delta}{A}$

  \end{itemize}
\end{definition}

Note that $P_{(\context{\tau}{A})}$ is equal to $\V(\Gamma,\tau) \times P_{A}$.

\begin{lemma}[fundamental lemma]
  Any logical predicate $\{P_{i}\}_{i\in I}$ for $\T$ straightforwardly defines a (set-theoretic) interpretation $\interpret{-}$.
  In particular, a derivable judgement $\context{\type{x}{\tau'}}{\Delta}\vdashc\type{N}{A}$ is interpreted as a function of the following form.
  \[
    \morphism{\interpret{\context{\type{x}{\tau'}}{\Delta}\vdashc\type{N}{A}}}{P_{(\context{\tau'}{\Delta})}}{P_{A}}
  \]
\end{lemma}

As a corollary, it follows that for any closed computation term $\type{N}{A}$, $P_{A}(N)$ holds.

\begin{remark}
  \label{rem:logical predicates at leftadj}
  We remark that $P_{\leftadj \tau}$ is not a set of terms of type $\leftadj \tau$ for the same reason of Remark \ref{rem:failure of completeness}.
  Namely, requiring $P_{\leftadj \tau}$ to be a set of terms of type $\leftadj \tau$ does not give rise to a strong equivariant functor.
  Due to this fact, we were unable to derive interesting syntactic results using this logical predicate.
  There are two ways to fix this: changing syntax or semantics.
  We expect whichever direction is hopeful, though we leave further investigation for our future work.
\end{remark}

\subsection{Actegorical gluing and \texorpdfstring{$\top\top$}{TT}-lifting}
\label{sec:actegorical gluing and TT-lifting}

The (categorical) $\top\top$-lifting \cite{DBLP:conf/csl/Katsumata05} is a technique to derive lifting of strong monads along a fibration.
It was originally introduced as a categorical formulation of logical predicates for the metalanguage \cite{DBLP:conf/tlca/LindleyS05}.
The basic idea of the $\top\top$-lifting comes with the following lemma.

\begin{lemma}
  \label{lem:fibred monad cod fibration}
  For a fibration $p: \E \to \B$, the projection functor $\pi: \Mnd(p) \to \Mnd(\B)$ is also a fibration, where $\Mnd(p)$ is the category of fibred monads over $p$.
\end{lemma}

To accomodate the continuation monad $S^{S^{(-)}}$ on $\E$, however, the lemma is insufficient because $S^{S^{(-)}}$ is not fibred even if $p$ strictly preserves the CCC structure.
The crucial ingredient of the $\top\top$-lifting was that it generalized Lemma \ref{lem:fibred monad cod fibration} by replacing $\Mnd(p)$ with $\Mnd'(p)$ the category of not-necessarily-fibred monads over $p$.
This is further generalized to (non-fibred) strong monads.
Consequently, given a preorder bifibration $p: \E \to \B$ preserving the CCC structure, a strong monad $T$ over $\B$, and some objects $S, R$ such that $S = TR$, the $\top\top$-lifting constructs a strong monad over $\E$ by the cartesian lifting of the canonical $\sigma$:
\begin{center}
  \begin{tikzcd}
    T^{\top\top} \ar[r, "\bar{\sigma}"] & S^{S^{(-)}} & \Mnd_{\text{strong}}'(p) \ar  [d, "\pi"] \\
    T \ar[r, "\sigma"'] & TR^{TR^{(-)}} & \Mnd_{\text{strong}}(\B)
  \end{tikzcd}  
\end{center}

We relate the $\top\top$-lifting to actegorical gluing in the following sense.

\begin{proposition}
  \label{prop:equivariant fibration from monoidal}
  Let $(T, \tilde{T})$ be a fibred monad over a fibration $p: \E \to \B$.
  Its Kleisli resolution gives us another fibration $p_T$ and a pullback diagram in $\Cat$.
  \begin{center}
    \begin{tikzcd}
      \E \ar[r, "\tilde{J}"] \ar[d, "p"'] \ar[loop left, "\tilde{T}"] \ar[rd, phantom, very near start, "\lrcorner"] & \E_{\tilde{T}} \ar[d, "p_T"] \\
      \B \ar[r, "J"'] \ar[loop left, "T"] & \B_{T}
    \end{tikzcd}    
  \end{center}
  If $\E$ and $\B$ are monoidal and $p$ is strict monoidal, and if $T$ and $\tilde{T}$ are both strong monads such that $\tilde{T}$'s strength $\tilde{t}$ is above $T$'s strength $t$, then both $\B_T$ and $\E_{\tilde{T}}$ have an $\E$-action and $p_T$ strictly preserves it.
\end{proposition}

Note that the $\E$-actions of $\B$ and $\B_T$ are given by change of base along $p$.
Therefore, if $p$ has a monoidal reflection (the prototypical example is the subobject fibration $\sub{\Sets} \to \Sets$), we can perform change of base along the reflection on the whole diagram above, which in turn allows us to give $\B$-actions to $\E$ and $\E_T$ and recover the original $\B$-actions for $\B$ and $\B_T$.

Proposition \ref{prop:equivariant fibration from monoidal} states that if $T^{\top\top}$ obtained by the $\top\top$-lifting is fibred, there exists a fibration that is strict equivariant, which yields $T^{\top\top}$ by gluing along $J$.
Although $T^{\top\top}$ is not fibred in general, some $T^{\top\top}$s that naturally arise in the semantics are indeed fibred.

For the subobject fibration $\sub{\Sets} \to \Sets$, we can calculate its $\top\top$-lifting as follows:
\[
  T^{\top\top}(P \subseteq X) \defeq \biggl(\biggl\{
    m \in TX
    \biggmid
    \forall c \in TR^X.
    (\forall x. P(x) \to S(c(x)))
    \to S(c^{\#}(m))
  \biggr\} \subseteq TX\biggr)
\]
where $c^{\#}$ is the Kleisli lifting of $c$.
The following configurations yield fibred monads.

\begin{enumerate}
  \item
    \emph{Exception.} When $T(X) \defeq X \uplus E$ and $(S \subseteq TR) \defeq (\{\ast\} \subseteq T(\{\ast\}))$ for some $E$, $T^{\top\top}$ is computed by $T^{\top\top}(P \subseteq X) = (P \subseteq X \uplus E)$.
  \item
    \emph{Partiality.} This case is subsumed by the exception monad where $E = \{\bot\}$.
  \item
    \emph{Nondeterminism.} When $T(X) \defeq \powerset_{\mathrm{fin}}(X)$ and $(S \subseteq TR) \defeq (\{\emptyset\} \subseteq T(\emptyset))$, $T^{\top\top}$ is computed by $T^{\top\top}(P \subseteq X) = (\{ m \in \powerset_{\mathrm{fin}}(X) \mid \exists x \in m. P(x) \} \subseteq \powerset_{\mathrm{fin}}(X))$.
\end{enumerate}

Nonexamples include the state monad and the continuation monad.
Proposition \ref{prop:equivariant fibration from monoidal} gives us another view of logical predicates of the metalanguage.
Such a view is sometimes more direct.
In the case of the exception monad, $\sub{\Sets}_{T^{\top\top}}$ has the following structure.
\begin{center}
  \AXC{$(P \subseteq X) \xrightarrow{f} (Q \subseteq Y)$ in $\sub{\Sets}_{T^{\top\top}}$}
  \doubleLine
  \UIC{$X \xrightarrow{f} Y $ in $\Sets$ such that $P(x)$ implies $\begin{cases}
    Q(f(x)) & (f(x) \in Y) \\
    \text{false} & (f(x) \in E)
  \end{cases}$}
  \DP  
\end{center}
Instantiating this to the case when $E = \{\bot\}$, the condition at the bottom may be viewed as the partial correctness of Hoare logic, by identifying $X$ and $Y$ as state (sub)spaces.
\begin{center}
  \AXC{$\morphism{f}{X}{Y}$ is in $\sub{\Sets}_{T^{\top\top}}(P, Q)$}
  \doubleLine
  \UIC{$\vdash_{\text{partial}} \{P\}\,f\,\{Q\}$}
  \DP  
\end{center}

\section{Concluding remarks}

This paper presented a new calculus of semi-effects (SEC) and its categorical models.
As an application of our semantics, we introduced actegorical gluing and derived logical predicates for the calculus.
A brief comparison with the $\top\top$-lifting is also presented.
SEC incorporates a more general notion of effects, as exemplified with $\mathtt{Applicative}$.
Unlike related work, our semantics is purely defined in terms of actegories and equivariant functors.

As pointed out in Remark \ref{rem:failure of completeness}, \ref{rem:failure of CHL}, and \ref{rem:logical predicates at leftadj}, we do not consider our semantics is fully satisfactory.
We expect that the true semantics of \sec{} is in the middle of lax equivariant and strong equivariant.
However, we do not know whether as clean an account as the present work is possible in this direction.
Another direction worth studying would be higher-order extensions.
All our development took place in a first-order setting.
While SEC's value side can be easily extended to a higher-order language, solely extending the computation side with higher-order functionals is not justified by the semantics, as Kleisli categories usually do not inherit a closed structure.
We would also need to find examples that are not supported by \Applicative{} but useful in practice.
Such examples might help to introduce to existing functional programming languages a new general framework for structuring programs.


\newpage
\appendix

\section{Omitted definitions}

\begin{definition}[strong monad]
  Let $\C$ be a category with finite products.
  A strong monad over $\C$ is a quadruple $(T, \eta, \mu, t)$ of a functor $\functor{T}{\C}{\C}$ and
  natural transformations $\morphism{\eta}{\id_{\C}}{T}$, $\morphism{\mu}{T^2}{T}$ and $\morphism{t}{(-)\times T(-)}{T(-\times -)}$ such that $(T, \eta, \mu)$ is a monad over $\C$ and the following diagrams commute.
  \begin{center}
   \begin{tikzpicture}
      \matrix[row sep=1cm, column sep=.7cm]{
        & \node (1TA) {$1\times TA$}; \\
       \node (T1A) {$T(1\times A)$}; & & \node (TA) {$TA$};\\
      };
      \draw 
        (1TA) edge node [above right] {$r_{TA}$} (TA)
              edge node [above left] {$t_{1, A}$} (T1A.25)
        (T1A) edge node [below] {$Tr_A$} (TA);
    \end{tikzpicture}
    \begin{tikzpicture}
      \matrix[row sep=1cm, column sep=1cm]{
        \node (ABTCl) {$(A\times B)\times TC$}; & & \node (ABTCr) {$A\times (B\times TC)$};\\
        \node (TABCl) {$T((A\times B)\times C)$}; & & \node (ATBC) {$A\times T(B\times C)$};\\
                                        & \node (TABCb) {$T(A\times (B\times C))$};\\
      };
      \draw 
        (ABTCl) edge node [above] {$\alpha_{A, B, TC}$} (ABTCr)
                edge node [left] {$t_{A\times B, C}$} (TABCl)
        (ABTCr)  edge node [right] {$\id_A\times t_{B, C}$} (ATBC)
        (TABCl)  edge node [below left] {$T\alpha_{A, B, C}$} (TABCb)
        (ATBC)  edge node [below right] {$t_{A, B\times C}$} (TABCb);
    \end{tikzpicture}
    \begin{tikzpicture}
      \matrix[row sep=1cm, column sep=1cm]{
        & \node (AB) {$A\times B$};\\
        \node (ATB) {$A\times TB$}; & & \node (TAB) {$T(A\times B)$};\\
        \node (ATTB) {$A\times T^2B$}; & \node (TATB) {$T(A\times TB)$}; & \node (TTAB) {$T^2(A\times B)$};\\
      };
      \draw 
        (AB) edge node [above left] {$\id_A\times \eta_B$} (ATB)
             edge node [above right] {$\eta_{A\times B}$} (TAB)
        (ATB) edge node [auto] {$t_{A, B}$} (TAB)
        (ATTB) edge node [left] {$\id_A\times \mu_B$} (ATB)
               edge node [below] {$t_{A, TB}$} (TATB)
        (TATB) edge node [below] {$Tt_{A, B}$} (TTAB)
        (TTAB) edge node [right] {$\mu_{A, B}$} (TAB);
    \end{tikzpicture}
  \end{center}
  It is straightforward to generalize the above definition to any monoidal category.
\end{definition}

\begin{definition}
  The full typing rules for \sec{} including finite product types.
  \begin{typingrules}
    \AxiomC{$\mathstrut$}\RightLabel{$(\type{x}{\tau})\in\Gamma$}\UnaryInfC{$\Gamma\vdashv\type{x}{\tau}$}\DisplayProof\quad
    \AxiomC{$\Gamma\vdashv\type{M_1}{\tau_1}$}\AxiomC{$\cdots$}\AxiomC{$\Gamma\vdashv\type{M_n}{\tau_n}$}\RightLabel{$\morphism{f}{\vec{\tau_i}}{\tau}$}\TrinaryInfC{$\Gamma\vdashv \type{f(M_1,\ldots,M_n)}{\tau}$}\DisplayProof\quad
    \AxiomC{$\Gamma\vdashv \type{M}{\tau}$}\AxiomC{$\Gamma\vdashv \type{M'}{\tau'}$}\BinaryInfC{$\Gamma\vdashv\type{\pair{M}{M'}}{\tau\times\tau'}$}\DisplayProof\quad
    \AxiomC{$\Gamma\vdashv\type{M}{\tau\times\tau'}$}\UnaryInfC{$\Gamma\vdashv\type{\proj1{M}}{\tau}$}\DisplayProof\quad
    \AxiomC{$\Gamma\vdashv\type{M}{\tau\times\tau'}$}\UnaryInfC{$\Gamma\vdashv\type{\proj2{M}}{\tau'}$}\DisplayProof\quad
    \AxiomC{$\mathstrut$}\UnaryInfC{$\Gamma\vdashv\type{\unitterm}{1}$}\DisplayProof\quad
    \AxiomC{$\mathstrut$}\UnaryInfC{$\context{\Gamma}{\type{v}{A}}\vdashc\type{v}{A}$}\DisplayProof\quad
    \AxiomC{$\Gamma\vdashv\type{M_1}{\tau_1}$}\AxiomC{$\cdots$}\AxiomC{$\Gamma\vdashv\type{M_n}{\tau_n}$}\AxiomC{$\context{\Gamma}{\Delta}\vdashc\type{N}{A}$}\RightLabel{$\morphism{g}{\vec\tau_i,A}{A'}$}\QuaternaryInfC{$\context{\Gamma}{\Delta}\vdashc \type{g(M_1,\ldots,M_n,N)}{A'}$}\DisplayProof\quad
    \AxiomC{$\Gamma\vdashv\type{M_1}{\tau_1}$}\AxiomC{$\cdots$}\AxiomC{$\Gamma\vdashv\type{M_n}{\tau_n}$}\RightLabel{$\morphism{h}{\vec\tau_i}{A'}$}\TIC{$\context{\Gamma}{\Delta}\vdashc \type{h(M_1,\ldots,M_n)}{A'}$}\DisplayProof\quad
    \AxiomC{$\Gamma\vdashv\type{M}{\tau}$}\UnaryInfC{$\context{\Gamma}{{\cdot}}\vdashc\type{\val{M}}{\leftadj\tau}$}\DisplayProof\quad
    \AxiomC{$\context{\Gamma}{\Delta}\vdashc\type{N_1}{\leftadj\tau}$}\AxiomC{$\context{\type{x}{\tau},\Gamma}{{\cdot}}\vdashc\type{N_2}{A}$}
    \BinaryInfC{$\context{\Gamma}{\Delta}\vdashc \type{\letvalin{x}{N_1}{N_2}}{A}$}\DisplayProof
  \end{typingrules}
\end{definition}

\begin{definition}
 Substitution of terms by a variable.
 \begin{align*}
   x_j[\vec M_i/\vec x_i] &\defeq M_j\\
   y[\vec M_i/\vec x_i] &\defeq y \qquad (\text{$x_j \ne y$ for any $j$})\\
   (f(M_1,\ldots,M_n))[M/x] &\defeq f(M_1[M/x],\ldots,M_n[M/x])\\
   \pair{M_1'}{M_2'}[\vec M_i/\vec x_i] &\defeq \pair{M_1'[\vec M_i/\vec x_i]}{M_2'[\vec M/\vec x]}\\
   \proj1{M'}[\vec M/\vec x] &\defeq \proj1{M'[\vec M/\vec x]}\\
   \proj2{M'}[\vec M/\vec x] &\defeq \proj2{M'[\vec M/\vec x]}\\
   \unitterm[\vec M/\vec x]  &\defeq \unitterm\\
   (\val{M'})[\vec M/\vec x] &\defeq \val{(M'[\vec M/\vec x])}\\
   (\letvalin{y}{N_1}{N_2})[\vec M/\vec x] &\defeq \letvalin{y}{N_1[\vec M/\vec x]}{N_2[\vec M/\vec x]}\\
   g(M'_1,\ldots,M'_n,N)[\vec M/\vec x] &\defeq g(M_1'[\vec M/\vec x],\ldots,M_n'[\vec M/\vec x],N[\vec M/\vec x]) \\
   h(M'_1,\ldots,M'_n)[\vec M/\vec x] &\defeq h(M_1'[\vec M/\vec x],\ldots,M_n'[\vec M/\vec x])
 \end{align*}
\end{definition}

\begin{definition}
  The inference rules for equations-in-context.
  The rules for congruence, reflectivity, symmetry, transitivity, and substitution are omitted.
  \begin{inferencerules}
        \AxiomC{$(\Gamma\vdashv M_1=_{\tau}M_2)\in \Ax$}\UnaryInfC{$\Gamma\vdashv M_1=_{\tau}M_2$}\DisplayProof\quad
        \AxiomC{$(\context{\Gamma}{\Delta}\vdashc N_1=_{A}N_2)\in \Ax$}\UnaryInfC{$\context{\Gamma}{\Delta}\vdashc N_1=_{A}N_2$}\DisplayProof\quad
        \AxiomC{$\context{\Gamma}{\Delta}\vdashc\type{N_1}{\leftadj\tau}$}\AxiomC{$\context{\type{x}{\tau}, \Gamma}{\emptycontext}\vdashc\type{N_2}{A}$}\RightLabel{\normalfont\rmfamily(comm.\,conv.)}
        \BinaryInfC{$\context{\Gamma}{\Delta}\vdashc{}\letvalin{x_1}{C[N_1]}{N_2}=_{A}C[\letvalin{x_1}{N_1}{N_2}]$}\DisplayProof\quad\\[0.2\baselineskip]
        \AxiomC{$\Gamma\vdashv\type{M}{\tau}$}\AxiomC{$\context{\type{x}{\tau},\Gamma}{\emptycontext}\vdashc\type{N}{A}$}\RightLabel{$(\beta)$}\BinaryInfC{$\context{\Gamma}{\emptycontext}\vdashc (\letvalin{x}{\val{M}}{N})=_{A} N[M/x]$}\DisplayProof\quad
        \AxiomC{$\context{\Gamma}{\Delta}\vdashc\type{N}{\leftadj\tau}$}\RightLabel{$(\eta)$}\UnaryInfC{$\context{\Gamma}{\Delta}\vdashc (\letvalin{x}{N}{\val{x}})=_{\leftadj\tau}N$}\DisplayProof
      \end{inferencerules}
\end{definition}

\begin{definition}[internal language]
  Let $\functor{F}{\V}{\C}$ be a model of \sec{} that is ``small'', i.e.\ $\V$ and $\C$ are small.
  The internal language of $F$ is given by the following data.

  \emph{Signature.} Let $\Sigma_F$ be a signature of \sec{} such that
  \begin{enumerate}
    \item
      (base types.) Base types are given by the sets of objects.
      The set of base value types is $\ob\V$ and the set of base computation type is $\ob\C$.
      We write $\enc{\tau}$ for the type corresponding to $\tau \in \V$, and $\enc{A}$ for the type of $A \in \C$.

      Note that fixing these two defines the interpretation of all types and contexts.

    \item
      (function symbols.) Functions symbols are given by the sets of morphisms.
      Explicitly, we use the following set as the set of function symbols.
      \[
        \{ \vec{\tau} \xrightarrow{\enc{f}} \tau \mid \sem{\vec{\tau}} \xrightarrow{f} \sem{\tau} \}
        \cup \{ \vec{\tau},A \xrightarrow{\enc{g}} A' \mid \sem{\vec{\tau}} \cdot \sem{A} \xrightarrow{g} \sem{A'} \}
        \cup \{ \vec{\tau} \xrightarrow{\enc{h}} A \mid F\sem{\vec{\tau}} \xrightarrow{h} \sem{A} \}
      \]
  \end{enumerate}
  $\Sigma_F$ has a canonical interpretation $\sem{\cdot}$ in $F$.

  \emph{Theory.} The internal language $\T_F$ is a theory over $\Sigma_F$ given by the following set of axioms:
  \[
    (L_1 = L_2) \in \Ax \iff \sem{L_1} = \sem{L_2} \text{ in $F$}
  \]
\end{definition}

\begin{definition}
  Let $\functor{F}{\V}{\C}, \functor{F'}{\V}{\C'}$ be models of \sec{} sharing the domain $\V$.
  Let $(H, \theta), (H', \theta')$ be morphisms from $F$ to $F'$.
  A 2-cell $\morphism{\alpha}{(H, \theta)}{(H', \theta')}$ is a 2-cell in $\Act{\V}$ from $H$ to $H'$ subject to the following equation.
  \[
    \begin{tikzpicture}[baseline=(arrow.base)]
      \matrix[row sep=1cm, column sep=1cm]{
        & \node (V) {$\V$};\\
        \node (C){$\C$}; & & \node (s) {$\C'$};\\
      };
      \coordinate (empty) at ($(V) !.6! (s)$);
      \draw 
        (V) edge node [above left] {$F$} (C)
            edge node [above right] {$F'$} (s)
        (C) edge node (a) [below, align=center] {$H$} (s)
            edge[bend right=60] node [below] {H'} (s);
  
      \node (arrow)
        at ($(C)!.7!(empty)$) [inner sep=0mm, rotate=30, label=$\theta$] {$\Leftarrow$};
      \node (cell) at (0,-1.4) [label=$\alpha$, rotate=90] {$\Leftarrow$};
    \end{tikzpicture}
    =
    \begin{tikzpicture}[baseline=(arrow.base)]
      \matrix[row sep=1cm, column sep=1cm]{
        & \node (V) {$\V$};\\
        \node (C){$\C$}; & & \node (s) {$\C'$};\\
      };
      \coordinate (empty) at ($(V) !.6! (s)$);
      \draw 
        (V) edge node [above left] {$F$} (C)
            edge node [above right] {$F'$} (s)
        (C) edge node (a) [below, align=center] {$H'$} (s);
  
      \node (arrow)
        at ($(C)!.7!(empty)$) [inner sep=0mm, rotate=30, label=$\theta'$] {$\Leftarrow$};
    \end{tikzpicture}
  \]
\end{definition}

\section{Omitted proofs}

\begin{proof}[Proof of Theorem \ref{thm:model by ef}]
  The nontrivial point is to check that $\interpret{-}$ is sound with respect to ($\beta$), ($\eta$), and (comm.\,conv.).
  We check only $(\beta)$ and $(\eta)$.
  For $(\beta)$, see the following diagram.
  \begin{center}
    \begin{tikzpicture}
      \matrix[row sep=1.5cm, column sep=2cm]{
        \node (lg) {$\leftadj\interpret{\Gamma}$}; & \node (lgg) {$\leftadj (\interpret{\Gamma}\times\interpret{\Gamma})$}; & \node (glg) {$\interpret{\Gamma}\cdot\leftadj\interpret{\Gamma}$};\\
                                        & \node (lgt) {$\leftadj (\interpret{\Gamma}\times\interpret{\tau})$}; & \node (glt) {$\interpret{\Gamma}\cdot\leftadj\interpret{\tau}$};\\
                                        & & \node (lgt2) {$\leftadj (\interpret{\Gamma}\times\interpret{\tau})$}; \\
                                        & & \node (A) {$\interpret{A}$};\\
      };
      \draw 
        (lg)  edge node [above] {$\leftadj\delta$} (lgg)
              edge node [below left] {$\leftadj\pair{\id}{\interpret{M}}$} (lgt)
        (lgg) edge node [above] {$\inverse\phi$} (glg)
              edge node [right] {$\leftadj(\id\times\interpret{M})$} (lgt)
        (glg) edge node [right] {$\id\cdot\leftadj\interpret{M}$} (glt)
        (lgt) edge node [below] {$\inverse{\phi}$} (glt)
              edge [equal] (lgt2)
        (glt) edge node [right] {$\phi$} (lgt2)
        (lgt2) edge node [right] {$\interpret{N}$} (A);
      \node at (-2, 2.5) {(a)};
      \node at (2.5, 2) {(b)};
    \end{tikzpicture}
  \end{center}
  In the above diagram, the small diagram labeled (a) commutes by the universality of 
  the product and the one labeled (b) does by the naturality of $\phi$.

  For $(\eta)$, we show only the case where $\Delta$ is nonempty because the case where $\Delta$ is empty can be
  shown in essentially the same way. See the following diagram. 
  \begin{center}
    \begin{tikzpicture}
      \matrix[row sep=.15cm, column sep=.2cm]{
        \node (GD1) {$\interpret{\Gamma}\cdot\interpret{\Delta}$}; & & & & & &  \node (GGD1) {$(\interpret{\Gamma}\times\interpret{\Gamma})\cdot \interpret{\Delta}$};\\
        & & \node [label=above right:(a)]{};\\
        & & \node (1GD1) {$(1\times\interpret{\Gamma})\cdot\interpret{\Delta}$};\\
        & \node {(b)}; & & & \node {(c)};\\
        & & \node (1GD2) {$1\cdot (\interpret{\Gamma}\cdot\interpret{\Delta})$}; & & & & \node (GGD2) {$\interpret{\Gamma}\cdot(\interpret{\Gamma}\cdot\interpret{\Delta})$};\\
        & & & & \node {(e)};\\
        & & \node (1lt) {$1\cdot\leftadj\interpret{\tau}$}; & & & & \node (Glt) {$\interpret{\Gamma}\cdot\leftadj\interpret{\tau}$};\\
        & & & & \node {};\\
        & \node {};&  & \node {}; & \node (l1t) {$\leftadj(1\times\interpret{\tau})$}; & &  \node (lGt) {$\leftadj(\interpret{\Gamma}\times\interpret{\tau})$};\\
        & & & & \node {};\\
        \node (GD2) {$\interpret{\Gamma}\cdot\interpret{\Delta}$}; & & \node (lt1) {$\leftadj\interpret{\tau}$}; & & & & \node (lt2) {$\leftadj\interpret{\tau}$};\\
      };
      \draw 
        (GD1) edge [equal] (GD2)
              edge node [above] {$\delta\cdot\id$} (GGD1)
              edge node [above right] {$\inverse l\cdot \id$} (1GD1)
        (GGD1) edge node [right] {$\mu$} (GGD2)
               edge node [below right] {$(\bang\times\id)\cdot\id$} (1GD1)
               (1GD1) edge node [right] {$\mu$} (1GD2)
        (1GD2) edge node [above left] {$\eta$} (GD2)
               edge node [right] {$\id\cdot\interpret{M}$} (1lt)
        (GGD2) edge node [below] {$\bang\cdot\id$} (1GD2)
               edge node [right] {$\id\cdot\interpret{M}$} (Glt)
        (1lt)  edge node [right] {$\eta$} (lt1)
               edge node [above right] {$\alpha$} (l1t)
        (Glt)  edge node [below] {$\bang\cdot\id$} (1lt)
               edge node [right] {$\alpha$} (lGt)
        (l1t)  edge node [below right=-2mm] {$\leftadj\pi'$} (lt1)
        (lGt)  edge node [below] {$\leftadj(\bang\times\id)$} (l1t)
               edge node [right] {$\leftadj\pi'$} (lt2)
        (GD2)  edge node [below] {$\interpret{M}$} (lt1)
        (lt1)  edge [equal] (lt2);
      \node at (-2.8, -2.5) {(d)};
      \node at (-.5, -2.5) {(g)};
      \node at (2.8, -2) {(f)};
      \node at (1.5, -3.2) {(h)};
    \end{tikzpicture}
  \end{center}
  In the above diagram, each small diagram commutes because of:
  \begin{enumerate}
    \renewcommand{\labelenumi}{(\alph{enumi})}
    \item 
      the functoriality of the monoidal action and the fact that 
      $l\comp (\bang\times\id)\comp\delta=\id$,
    \item 
      the coherence for monoidal actions,
    \item 
      the naturality of $\mu$ and the functoriality of the monoidal action,
    \item 
      the naturality of $\eta$,
    \item 
      the functoriality of the monoidal action,
    \item 
      the naturality of $\alpha$,
    \item 
      the coherence for $\alpha$ and
    \item 
      the fact that $\pi'\comp (\bang\times\id) =\pi'$ and the functoriality of $\leftadj$.
  \end{enumerate}

  The calculation for (comm.\,conv.) is more complicated and the involved diagram gets bigger than those for the above two,
  but they still are straightforward, and we omit them.
\end{proof}

\begin{proof}[Proof of Example \ref{ex:model of sec}.\ref{ex:model from lax monoidal}]
  Let $\C$ be a cartesian closed category and $F \colon \C \to \C$ a lax monoidal functor.
  Let $\D$ be the category defined by the following data:
  \begin{align*}
    \ob\D &\defeq \ob\C \\
    \D(X,Y) &\defeq \C(1, F(X \Rightarrow Y)) \\
    \left(X \xrightarrow{\id_X} X\right) &\defeq \left(1 \xrightarrow{\iota} F1 \xrightarrow{F(\overline{\id_X})} F(X \Rightarrow X)\right) \\
    \left(Y \xrightarrow{g} Z\right) \circ \left(X \xrightarrow{f} Y\right) &\defeq
        1 \cong 1 \times 1 \\
        &\quad \xrightarrow{f \times g} F(X \Rightarrow Y) \times F(Y \Rightarrow Z) \\
        &\quad \xrightarrow{\mu} F((X \Rightarrow Y) \times (Y \Rightarrow Z)) \\
        &\quad \xrightarrow{F(\mathtt{comp})} F(X \Rightarrow Z)
  \end{align*}
  where $\iota$ and $\mu$ are the morphisms required by the lax monoidality of $F$, $X \Rightarrow Y$ is the exponent in $\C$, and $\mathtt{comp}$ is the moprhism $\lambda \langle f, g \rangle. \lambda x. g(f(x))$.
  In an abstract view, $\D$ is described in terms of enriched categories.
  Given a (symmetric) monodal closed category $\M$, there exists a 2-category $\ECat{\M}$ of $\M$-enriched categories.
  Similarly, given (symmetric) monodal closed categories $\M,\M'$ and a lax monoidal functor $K \colon \M \to \M'$, there exists a 2-functor $K_\ast : \ECat{\M} \to \ECat{\M'}$ defined by change-of-base along $K$.
  We can then define $\D$ by two successive applications of change-of-base ${\C(1,-)_\ast F_\ast \C}$, where $\C(1,-) \colon \C \to \Sets$ is the global section functor.

  $\D$ has an $\C$-action defined by finite products:
  \begin{align*}
    X \cdot Y &\defeq X \times Y \\
    \left(X \xrightarrow{f} X'\right) \cdot \left(Y \xrightarrow{g} Y'\right) &\defeq 
      1 \cong 1 \times 1 \\
      &\quad\xrightarrow{F\overline{f} \times g} F(X \Rightarrow X') \times F(Y \Rightarrow Y') \\
      &\quad\xrightarrow{\mu} F((X \Rightarrow X') \times (Y \Rightarrow Y')) \\
      &\quad\xrightarrow{F(\mathtt{prod})} F(X \times Y \Rightarrow X' \times Y')
  \end{align*}
  where $\mathtt{prod}$ is the morphism $\lambda \langle f, g \rangle. \lambda \langle x, y \rangle. \langle f x, g y \rangle$.

  Then there exists an identity-on-object functor $J \colon \C \to \D$, whose action on morphisms is given by:
  \[
    J(X \xrightarrow{f} Y) \defeq \left(1 \xrightarrow{\iota} F1 \xrightarrow{F\overline{f}} F(X \Rightarrow Y)\right).
  \]

  We show that $J$ is strict equivariant, for which $\phi^J_{X,Y} \colon X \cdot J(Y) \to J(X \times Y)$ is given by the identity.
  Here we only check that $\phi^J$ is natural in both variables, by chasing the following diagrams.
  \[
    \begin{tikzcd}
      1 \ar[rr, "\iota"] \ar[d, "\cong"'] && F1 \ar[ddd, "F(\overline{f \times g})"] \ar[dl, "\cong"] \ar[ddl, "\cong"] \\
      1 \times 1 \ar[d, "\iota \times \iota"'] \ar[r, "\id \times \iota"] & 1 \times F1 \ar[ld, "\iota \times \id"'] \\
      F1 \times F1 \ar[d, "F\overline{f} \times F\overline{g}"'] \ar[r, "\mu"] & F(1 \times 1) \ar[d, "F(\overline{f} \times \overline{g})"] \\
      F(X \Rightarrow X') \times F(Y \Rightarrow Y') \ar[r, "\mu"'] & F((X \Rightarrow X') \times (Y \Rightarrow Y')) \ar[r, "\mathtt{prod}"'] & F(X \times Y \Rightarrow X' \times Y')
    \end{tikzcd}
  \]
  In the above diagram, the bottom left composite is $f \cdot Jg$ and the upper right is $J(f \times g)$.
\end{proof}

\begin{proof}[Proof of Theorem \ref{thm:theory to ef}]
  Let $\T$ be any theory of \sec.
  We define the monoidal category $\V_{\T}$, the category $\C_{\T}$ with a $\V_{\T}$-action and the strong equivariant functor $\morphism{F}{\V_{\T}}{\C_{\T}}$.

  Firstly we define the cartesian category $\V_{\T}$.
  $\V_{\T}$ has value types as its object and as its morphisms $\morphism{}{\tau_1}{\tau_2}$
  equivalence classes of value terms $\classof{\type{x}{\tau_1}\vdashv \type{M}{\tau_2}}_{\T}$ derived under $\T$. 
  We define the equivalence class by
  \begin{align*}
    &\classof{\type{x}{\tau_1}\vdashv\type{M}{\tau_2}}_{\T} = \classof{\type{x'}{\tau_1}\vdashv\type{M'}{\tau_2}}_{\T} \\
    \iff &\type{x}{\tau_1}\vdashv M=_{\tau_2}M'[x/x']\text{ is derived}.
  \end{align*}
  The identity morphism $\id_{\tau}$ is $\classof{\type{x}{\tau}\vdashv\type{x}{\tau}}_{\T}$.
  The composition $\classof{\type{x_2}{\tau_2}\vdashv \type{M_3}{\tau_3}}_{\T}\comp \classof{\type{x_1}{\tau_1}\vdashv \type{M_2}{\tau_2}}_{\T}$ is
  $\classof{\type{x_1}{\tau_1}\vdashv \type{M_3[M_2/x_2]}{\tau_3}}_{\T}$.
  We can show that $\V_{\T}$ has binary products and the terminal object.
  The binary products of $\tau_1$ and $\tau_2$ is $\tau_1\times\tau_2$ and the terminal object is $1$.

  We define $\C_{\T}$ next.
  $\C_{\T}$ has two kinds of objects:
  \begin{itemize}
    \item
      pairs $(\tau\mid A)$ of a value type $\tau$ and a computation type $A$, and 
    \item 
      value types $\tau$, which we write $(\context{\tau}{\emptycontext})$.
  \end{itemize}

  For morphisms, 
  we have to be careful.
  At first, we define the equivalence class of typed computation terms as we do to define the morphisms in $\V_{\T}$.
  \begin{align*}
    &\classof{\context{\type{x}{\tau}}{\type{v}{A}}\vdashc\type{N}{A'}}_{\T} = \classof{\context{\type{x'}{\tau}}{\type{v'}{A}}\vdashv\type{N'}{A'}}_{\T} \\
    \iff &\context{\type{x}{\tau}}{\type{v}{A}}\vdashc N=_{A'}N'[v/v', x/x']\text{ is derived} \\[.5\baselineskip]
    &\classof{\context{\type{x}{\tau}}{\emptycontext}\vdashc\type{N}{A'}}_{\T} = \classof{\context{\type{x'}{\tau}}{\emptycontext}\vdashv\type{N'}{A'}}_{\T} \\
    \iff &\context{\type{x}{\tau}}{\emptycontext}\vdashc N=_{A'}N'[x/x']\text{ is derived}
  \end{align*}
  In the sequel, the subscript $\T$ for the equivalence classes will be omitted for simplicity.

  To define what is a morphism $\morphism{}{(\tau\mid\Delta)}{(\tau\mid\Delta')}$,
  we have to be careful of whether $\Delta$ and $\Delta'$ are empty or not.
  \begin{itemize}
    \item When $\Delta'$ is not empty, say $\Delta'=A'$, a morphism $\morphism{}{(\context{\tau}{\Delta})}{(\context{\tau'}{A'})}$ is a pair of 
      equivalence classes of a typed term that has the following form.
      \[
        (\classof{\type{x}{\tau}\vdashv \type{M'}{\tau'}}, \classof{\context{\type{x}{\tau}}{\Delta}\vdashc\type{N}{A'}})
      \]

    \item When both $\Delta$ and $\Delta'$ are empty, 
      a morphism $\morphism{}{(\context{\tau}{\emptycontext})}{(\context{\tau'}{\emptycontext})}$ is an equivalence class of typed value terms which has the following form.
      \[
        \classof{\type{x}{\tau}\vdashv \type{M'}{\tau'}}
      \]

    \item When $\Delta$ isn't empty and $\Delta'$ is empty, 
      $\C_{\T}((\context{\tau}{\Delta}), (\context{\tau'}{\Delta'}))$ is empty.
  \end{itemize}
  The identity morphism $\morphism{\id}{(\tau\mid\emptycontext)}{(\tau\mid\emptycontext)}$ is $\classof{\type{x}{\tau}\vdashv\type{x}{\tau}}_{\T}$ and
  $\morphism{\id}{(\tau\mid A)}{(\tau\mid A)}$ is $(\classof{\type{x}{\tau}\vdashv\type{x}{\tau}}_{\T}, \classof{\context{\type{x}{\tau}}{\type{v}{A}}\vdashc \type{v}{A}})$.
  Next, we consider the composition of the morphism $\morphism{f}{(\tau\mid\Delta)}{(\tau'\mid\Delta')}$ and $\morphism{g}{(\tau'\mid\Delta')}{(\tau''\mid\Delta'')}$.
  Note that there are three cases to consider.
  \begin{itemize}
    \item When $\Delta$, $\Delta'$ and $\Delta''$ are empty. 
      Let
      \begin{align*}
        f&=\classof{\type{x}{\tau}\vdashv \type{M}{\tau'}}_{\T}\\
        g&=\classof{\type{x'}{\tau'}\vdashv \type{M'}{\tau''}}_{\T}, 
      \end{align*}
      then we define $g\comp f = \classof{\type{x}{\tau}\vdashv\type{M'[M/x']}{\tau''}}_{\T}$.
    \item When $\Delta$ and $\Delta'$ is empty and $\Delta''$ is not empty.
      Let 
      \begin{align*}
        f &= \classof{\type{x}{\tau}\vdashv \type{M'}{\tau'}}\\
        g &= (\classof{\type{x'}{\tau'}\vdashv \type{M''}{\tau''}}, \classof{\context{\type{x'}{\tau'}}{\emptycontext}\vdashc\type{N''}{A''}}),
      \end{align*}
      then we define
      \[
        g\comp f = (\classof{\type{x}{\tau}\vdashv \type{M''[M'/x']}{\tau''}}, \classof{\context{\type{x}{\tau}}{\emptycontext}\vdashc\type{N''[M'/x']}{A''}}).
      \]

    \item When $\Delta$ is empty and $\Delta'$ and $\Delta''$ is not empty.
      Let 
      \begin{align*}
        f &= (\classof{\type{x}{\tau}\vdashv \type{M'}{\tau'}}, \classof{\context{\type{x}{\tau}}{\Delta}\vdashc\type{N'}{A'}})\\
         &= (\classof{\type{x'}{\tau'}\vdashv \type{M''}{\tau''}}, \classof{\context{\type{x'}{\tau'}}{\type{v'}{A'}}\vdashc\type{N''}{A''}}),
      \end{align*}
      then we define
      \[
        g\comp f = (\classof{\type{x}{\tau}\vdashv \type{M''[M'/x']}{\tau''}}, \classof{\context{\type{x}{\tau}}{\Delta}\vdashc\type{N''[M'/x', N'/v']}{A''}})
      \]
  \end{itemize}

  Up to here, we define two categories $\V_{\T}$ and $\C_{\T}$. Next, we define a 
  monoidal action $\functor{(-)\cdot(-)}{\V_{\T}\times\C_{\T}}{\C_{\T}}$.
  For objects $\tau$ in $\V_{\T}$ and $(\tau'\mid\Delta)$ in $\C_{\T}$, we define $\tau\cdot (\tau'\mid\Delta) = (\tau\times\tau'\mid\Delta)$.
  For morphisms $\morphism{\classof{\type{x}{\tau_1}\vdashv\type{M}{\tau_2}}_{\T}}{\tau_1}{\tau_2}$ in $\V_{\T}$ and 
  $\morphism{(\classof{\type{x'}{\tau_3}\vdashv\type{M'}{\tau_4}}_{\T}, \classof{\context{\type{y}{\tau_3}}{\Delta}\vdashc\type{N}{A}}_{\T})}{(\tau_3\mid\Delta)}{(\tau_4\mid A)}$,
  we define 
  \begin{align*}
    & \classof{\type{x}{\tau_1}\vdashv\type{M}{\tau_2}}_{\T}\cdot (\classof{\type{x'}{\tau_3}\vdashv\type{M'}{\tau_4}}_{\T}, \classof{\context{\type{y}{\tau_3}}{\Delta}\vdashc\type{N}{A}}_{\T})\\
    ={} & 
      \begin{aligned}[t]
        (&\classof{\type{z}{\tau_1\times\tau_3}\vdashv\type{\pair{M[\proj1{z}/x]}{M'[\proj2{z}/x']}}{\tau_2\times\tau_4}}_{\T},\\
        &\classof{\context{\type{w}{\tau_1\times\tau_3}}{\Delta}\vdashc \type{N/[\proj2{w}/y]}{A}}_{\T}).
      \end{aligned}
  \end{align*}
  For morphisms $\morphism{\classof{\type{x}{\tau_1}\vdashv\type{M}{\tau_2}}_{\T}}{\tau_1}{\tau_2}$ in $\V_{\T}$ and 
  $\morphism{\classof{\type{x'}{\tau_3}\vdashv\type{M'}{\tau_4}}_{\T}}{(\tau_3\mid\emptycontext)}{(\tau_4\mid \emptycontext)}$, 
  the action is defined in a similar way.
  It is straightforward to show that $(-)\cdot(-)$ preserves identities and compositions.

  From the above, we can get a functor $\functor{(-)\cdot (-)}{\V_{\T}\times\C_{\T}}{\C_{\T}}$.
  In order to make this functor a $\V_{\T}$-action, 
  we have to define the coherence natural isomorphism $\morphism{\eta_{(\context{\tau}{\Delta})}}{1\cdot (\context{\tau}{\Delta})}{(\context{\tau}{\Delta})}$ and
  $\morphism{\mu_{\tau, \tau', (\context{\tau''}{\Delta})}}{(\tau\times\tau')\cdot(\context{\tau''}{\Delta})}{\tau\cdot(\tau'\cdot(\context{\tau''}{\Delta}))}$ for this monoidal action.

  As to $\eta_{(\context{\tau}{\Delta})}$, 
  \begin{itemize}
    \item if $\Delta$ is empty, we define $\eta_{(\context{\tau}{\emptycontext})}$ to be 
      $\classof{\type{x}{1\times\tau}\vdashv\type{\proj2{x}}{\tau}}_{\T}$, and
    \item if $\Delta$ isn't empty, say $\Delta=A$, we define $\eta_{(\context{\tau}{A})}$ to be 
      \[
      (\classof{\type{x}{1\times\tau}\vdashv\type{\proj2{x}}{\tau}}_{\T}, \classof{\context{\type{x}{1\times\tau}}{\type{v}{A}}\vdashc\type{v}{A}}).
      \]
  \end{itemize}
  As to $\mu_{(\context{\tau}{\Delta})}$, 
  \begin{itemize}
    \item if $\Delta$ is empty, we define $\mu_{\tau,\tau',(\context{\tau''}{\emptycontext})}$ to be 
      \[
        \classof{\type{x}{(\tau\times\tau')\times\tau''}\vdashv\type{\pair{\proj1{\proj1{x}}}{\pair{\proj2{\proj1{x}}}{\proj2{x}}}}{\tau\times(\tau'\times\tau'')}}_{\T},
      \]
      and
    \item if $\Delta$ isn't empty, say $\Delta=A$, we define $\mu_{\tau,\tau',(\context{\tau''}{A})}$ to be 
      \begin{align*}
        (&\classof{\type{x}{(\tau\times\tau')\times\tau''}\vdashv\type{\pair{\proj1{\proj1{x}}}{\pair{\proj2{\proj1{x}}}{\proj2{x}}}}{\tau\times(\tau'\times\tau'')}}_{\T},\\
         &\classof{\context{\type{x}{(\tau\times\tau')\times\tau''}}{\type{v}{A}}\vdashc\type{v}{A}}).
      \end{align*}
  \end{itemize}
  The coherence condition for these natural isomorphism reduces to the 
  monoidality of $\V_{\T}$.

  Next, We define an equivariant functor $\functor{F}{\V_{\T}}{\C_{\T}}$ and its coherent natural transformation $\morphism{\alpha}{(-)\cdot F(-)}{F((-)\cdot(-))}$.
  For object $\tau$ in $\V_{\T}$, we define $F\tau=(\context{\tau}{\emptycontext})$ and
  $F$ is identity on morphisms. $F$ is clearly a functor. Moreover, $F$ is equivariant with
  an identity natural transformation. 
\end{proof}

\begin{proof}[Proof of Example \ref{ex:2-category of actegories}]
  We show that an equivariant left adjoint is always strong equivariant, which we believe is folklore.
  Notice the similarity with a fact about monoidal functors: a monoidal left adjoint is always strong monoidal.

  Let $F \dashv G$ an $\M$-equivariant adjunction.
  We define $\psi^F \colon F(m \cdot x) \to m \cdot F(x)$ to be the mate of $m \cdot x \xrightarrow{m \cdot \eta} m \cdot GFx \xrightarrow{\phi^G} G(m \cdot Fx)$.
  Check the following diagrams to see that $\psi^F$ is an inverse of $\phi^F$.
  \[
    \begin{tikzcd}
      F(m \cdot x) \ar[r, "F(m \cdot \eta)"] \ar[rrd, "F\eta"'] \ar[rrrd, bend right, "\id"'] \ar[rrr, bend left, "\psi^F"] & F(m \cdot GFx) \ar[r, "F\phi^G"] \ar[rd, "F\phi^{GF}"] & FG(m \cdot Fx) \ar[r, "\epsilon"] \ar[d, "FG\phi^F"] & m \cdot Fx \ar[d, "\phi^F"] \\
      && FGF(m \cdot x) \ar[r, "\epsilon"'] & F(m \cdot x)
    \end{tikzcd}
  \]
  \[
    \begin{tikzcd}
      m \cdot Fx \ar[d, "\phi^F"'] \ar[r, "m \cdot F\eta"] \ar[rrrd, bend left, "m \cdot \id"] & m \cdot FGFx \ar[d, "\phi^F"] \ar[rd, "\phi^{FG}"'] \ar[rrd, "m \cdot \epsilon"] \\
      F(m \cdot x) \ar[r, "F(m \cdot \eta)"'] \ar[rrr, bend right, "\psi^F"'] & F(m \cdot GFx) \ar[r, "F\phi^G"'] & FG(m \cdot Fx) \ar[r, "\epsilon"'] & m \cdot Fx
    \end{tikzcd}
  \]
\end{proof}

\begin{proof}[Proof of Remark \ref{rem:failure of CHL}]
  The Curry-Howard-Lambek correspondence usually refers to the following equivalence in a suitable 2-category.
  \[
    \M \xrightarrow{\simeq} \operatorname{Syn}(\operatorname{Lang}(\M))
  \]
  where $\M$ is a categorical model and $\operatorname{Syn}$ and $\operatorname{Lang}$ are operators giving the term model and the internal language.

  To adapt this to models of \sec{}, we slightly modify the term model $\operatorname{Syn}(\T)$ in Theorem \ref{thm:theory to ef} as follows:
  \[
    \operatorname{Syn}(\T)_F \colon \V \xrightarrow{\simeq} \operatorname{Syn}(\T)_\V \xrightarrow{(\ast)} \operatorname{Syn}(\T)_\C
  \]
  where $(\ast)$ is the term model presented in Theorem \ref{thm:theory to ef} and the first equivalence is the Curry-Howard-Lambek correspondence of algebraic theory (with finite products).

  Then the question is reduced to existence of the following equivalence.
  \[
    F \xrightarrow{\simeq} \operatorname{Syn}(\T_F)
  \]
  We want to make the following $(H,\theta)$ from $F$ to $\operatorname{Syn}(\T_F)$ the witness of the above equivalence.
  \begin{align*}
    H(A) &\defeq (\context{\emptyctx}{\enc{A}}) \\
    \theta_\tau \colon (\context{\enc{\tau}}{\emptyctx}) \to (\context{\emptyctx}{\enc{F\tau}}) &\defeq \classof{\context{\type{x}{\enc{\tau}}}{\emptyctx} \vdashc \type{\enc{\id_{F\tau}}(x)}{\enc{F\tau}}}
  \end{align*}
  Then the morphism in the reverse direction $(H',\theta')$ will be
  \begin{align*}
    H'((\context{\Gamma}{\Delta})) &\defeq \sem{\context{\Gamma}{\Delta}}_F \\
    \theta'_\tau \colon F\tau \to F\tau &\defeq \id_{F\tau}.
  \end{align*}
  If these form an equivalence, there should be a 2-cell $\alpha: H\circ H' \to 1$.
  However, this is impossible.
  For example, the $(\context{\tau}{\emptyctx})$ component of $\alpha$ has the following type:
  \[
    \alpha_{(\context{\tau}{\emptyctx})} \colon (\context{\emptyctx}{\enc{F\sem{\tau}}}) \to (\context{\tau}{\emptyctx}).
  \]
  By the definition of the term model, there is no such morphism.
\end{proof}

\begin{proof}[Proof of Proposition \ref{prop:simple gluing}]
  Define $m\action (D, C, f)$ to be $(m\action D, m\action C, \phi\comp (m\action f))$ for objects $m$ in $\M$ and $(D, C, \morphism{f}{D}{\Gamma C})$ in $\glue{\D}{\Gamma}$,
  and $a\action (d, c)$ to be $(a\action d, a\action c)$ for morphisms $\morphism{a}{m}{m'}$ in $\M$ and $\morphism{(d, c)}{(D, C, f)}{(D', C', f')}$ in $\glue{\D}{\Gamma}$.
  It follows that $(a\action d, a\action c)$ is a morphism 
  $\morphism{}{(m\action D, m\action C, \phi\comp(m\action f))}{(m'\action D', m'\action C', \phi\comp(m'\action f'))}$ in $\glue{\D}{\Gamma}$
  from the diagram below.
  \begin{center}
    \begin{tikzpicture}
      \matrix[row sep=1cm, column sep=1.8cm]{
        \node (md) {$m\action D$}; &
        \node (m'd')  {$m'\action D'$}; \\
        \node (mgc) {$m\action \Gamma C$}; &
        \node (m'gc') {$m'\action \Gamma C'$}; \\
        \node (gmc)  {$\Gamma(m\action C)$}; &
        \node (gm'c') {$\Gamma(m'\action C')$};\\
      };
      \draw 
        (md)  edge node [above] {$a\action d$} (m'd')
              edge node [left] {$m\action f$} (mgc)
        (mgc) edge node [above] {$a\action \Gamma c$} (m'gc')
              edge node [left] {$\phi$} (gmc)
        (gmc) edge node [below] {$\Gamma(a\action c)$} (gm'c')
        (m'd') edge node [right] {$m'\action f'$} (m'gc')
        (m'gc') edge node [right] {$\phi$} (gm'c')
      ;
    \end{tikzpicture}
  \end{center}
  The upper rectangle commutes by the fact that $(d, c)$ is a morphism in $\glue{\D}{\Gamma}$ and the (bi)fuctoriality of $(-)\action(-)$,
  and the lower one commutes by the naturality of $\phi$.
  It is straightforward to see $\functor{(-)\action(-)}{\M\times(\glue{\D}{\Gamma})}{\glue{\D}{\Gamma}}$
  is indeed an $\M$-action on $\glue{\D}{\Gamma}$.

  It is also straightforward to see that the projection functor $\functor{\pi}{\glue{\D}{\Gamma}}{\C}$ is 
  strict equivariant.
\end{proof}

\begin{proof}[Proof of Proposition \ref{prop:lax equivariant and opfibration}]
  The upper left category $\G$ has pairs $(C, X)$ of objects in $\C$ and $\E$ such that $\Gamma C=pX$ for its objects, 
  and pairs $(f, x)$ of morphisms in $\C$ and $\E$ such that $\Gamma f=px$ for its morphisms.
 
  We can define an $\M$-action on $\G$ by using universalities of opcartesian morphisms as follows.
  Let $m$ be any object in $\M$ and $(C, X)$ be any object in $\G$. 
  We define $m\cdot (C, X)$ to be $(m\cdot C, \opcleavage{(\phi_{m, C})}(m\cdot X))$.
  This definition is well-defined i.e.\ $\Gamma(m\cdot C) = p(\opcleavage{(\phi_{m, C})}(m\cdot X))$ holds.
  This follows from the fact that $\Gamma C=pX$ holds, $m\cdot pX=p(m\cdot X)$ holds because 
  $p$ is strict equivariant and $\phi_{m, C}$ has $m\cdot \Gamma C$ as its domain.
  \[
    \begin{array}{r@{{}}l@{{}}l}
      \simplemorphism{\phi_{m, c}&{}}{p(m\cdot X)&{}}{\Gamma(m\cdot C)}\\
      \simplemorphism{\oplifting{\phi_{m, c}}(m\cdot X)&{}}{m\cdot X &{}}{\opcleavage{(\phi_{m, C})}(m\cdot X)}
    \end{array}
  \]
  Let $\morphism{a}{m}{m'}$ be any moprhism of $\M$ and $\morphism{{(f, x)}}{(C, X)}{(C', X')}$ be any morphism of $\G$. 
  We define $a\cdot (f, x)$ by means of the universality of $\oplifting{\phi_{m, C}}(m\cdot X)$:
  we define $a\cdot (f, x)$ to be $(a\cdot f, u)$ where $u$ is the unique morphism which makes the left diagram commute and satisfies $pu=\Gamma(a\cdot f)$.
  Note that the lower right diagram commutes by the naturality of $\phi$, and so the upper right one does.
  \begin{center}
    \begin{tikzpicture}
      \matrix[row sep=1cm, column sep=1.5cm]{
        \node (mx) {$m\cdot X$}; & \node (Gmc) {$\opcleavage{(\phi_{m, C})}(m\cdot X)$};\\
        \node (m'x') {$m'\cdot X'$}; & \node (Gm'c') {$\opcleavage{(\phi_{m', C'})}(m'\cdot X')$};\\
      };
      \draw
        (mx) edge node [above=1mm] {$\oplifting{\phi_{m, C}}(m\cdot X)$} (Gmc)
             edge node [left] {$a\cdot x$} (m'x')
        (m'x') edge node [below=1mm] {$\oplifting{\phi_{m', C'}}(m'\cdot X')$} (Gm'c')
        (Gmc) edge [dashed] node [right] {$u$} (Gm'c');

      \node at (3.2, 0) {$\stackrel{p}{\mapsto}$};
      \node at (6.5, -2.2) [rotate=90] {$=$};
      \begin{scope}[xshift=6.5cm, yshift=-4cm]
        \matrix[row sep=1cm, column sep=1cm]{
          \node (mGc) {$m\cdot\Gamma C$}; & \node (Gmc) {$\Gamma (m\cdot C)$};\\
          \node (m'Gc') {$m'\cdot\Gamma C'$}; & \node (Gm'c') {$\Gamma (m'\cdot C')$};\\
        };
        \draw
          (mGc) edge node [above] {$\phi_{m, C}$} (Gmc)
               edge node [left] {$a\cdot \Gamma f$} (m'Gc')
          (m'Gc') edge node [below] {$\phi_{m', C'}$} (Gm'c')
          (Gmc) edge node [right] {$\Gamma(a\cdot f)$} (Gm'c');
      \end{scope}
      \begin{scope}[xshift=6.5cm]
        \matrix[row sep=1cm, column sep=2cm]{
          \node (pmx) {$p(m\cdot X)$}; & \node (Gmc) {$\Gamma (m\cdot C)$};\\
          \node (pm'x') {$p(m'\cdot X')$}; & \node (Gm'c') {$\Gamma (m'\cdot C')$};\\
        };
        \draw
          (pmx) edge node [above] {$\phi_{m, C}$} (Gmc)
                edge node [right] {$p(a\cdot x)$} (pm'x')
                (pm'x') edge node [below=1mm] {$p(\oplifting{\phi_{m', C'}}(m'\cdot X'))$} (Gm'c')
          (Gmc) edge node [right] {$\Gamma(a\cdot f)$} (Gm'c');
      \end{scope}
    \end{tikzpicture}
  \end{center}
  The functoriality of this monoidal action follows from the universalities of $\oplifting{\phi_{m, C}}$s.

  At last, we define the coherent natural isomorphisms for this monoidal action.
  Let $(C, X)$ be any object in $\G$. Consider the following diagram.
  \begin{center}
    \begin{tikzpicture}
      \matrix[row sep=1cm, column sep=1cm]{
        \node (1gc) {$1\cdot \Gamma C$}; & & \node (g1c) {$\Gamma(1\cdot C)$};\\
                                         & \node (gc) {$\Gamma C$};\\
      };
      \draw
        (1gc) edge node [above] {$\phi_{1, C}$} (g1c)
              edge node [below left] {$\eta_{\Gamma C}$} (gc)
        (g1c) edge node [below right] {$\Gamma\eta_{C}$} (gc);
    \end{tikzpicture}
  \end{center}
  This diagram commutes because it is one of those for coherence for $\phi$.
  In addition, $\eta_{\Gamma C}=p\eta_{X}$ holds because $\Gamma C=pX$ holds and $p$ is strict equivariant.
  With the universality of $\oplifting{\phi_{1, C}}$, the unique morphism $u$ which makes
  the following diagram commute and satisfies $pu=\Gamma\eta_{C}$ is obtained.
  \begin{center}
    \begin{tikzpicture}
      \matrix[row sep=1cm, column sep=1cm]{
        \node (1x) {$1\cdot X$}; & & \node (g1c) {$\opcleavage{(\phi_{1, C})}(1\cdot X)$};\\
                                         & \node (x) {$X$};\\
      };
      \draw
      (1x) edge node [above] {$\oplifting{\phi_{1, C}}(1\cdot X)$} (g1c)
              edge node [below left] {$\eta_{X}$} (x)
      (g1c.200) edge [dashed] node [below right] {$u$} (x);
    \end{tikzpicture}
  \end{center}
  We define $\eta_{(C, X)}$ to be $(\eta_{C}, u)$.

  In the same vein, $\mu_{m, m', (C, X)}$ is defined to be $(\mu_{m, m', C}, u)$ in which $u$ is the unique morphism which 
  makes the following diagram commute and satisfies $pu=\Gamma \mu_{m, m', C}$.
  \begin{center}
    \begin{tikzpicture}
      \matrix[row sep=1cm, column sep=4cm]{
        \node (mtmx) {$(m\tensor m')\cdot X$}; & \node (gmmc) {$\opcleavage{(\phi_{m\tensor m', C})}((m\tensor m')\cdot X)$};\\
        \node (mmx) {$m\cdot (m'\cdot X)$};\\
        \node (mgmc) {$m\cdot \opcleavage{(\phi_{m', C})}(m'\cdot X)$}; & \node (gmc) {$\opcleavage{(\phi_{m, m'\cdot C})}(m\cdot \opcleavage{(\phi_{m', C})}(m'\cdot X))$};\\
      };
      \draw 
        (mtmx) edge node [above] {$\oplifting{\phi_{m\tensor m', C}}((m\tensor m')\cdot X)$} (gmmc)
               edge node [left] {$\mu_{m, m', X}$} (mmx)
        (gmmc) edge [dashed] node [right] {$u$} (gmc)
        (mmx)  edge node [right] {$m\cdot \oplifting{\phi_{m', C}}(m'\cdot X)$} (mgmc)
        (mgmc) edge node [below=1mm] {$\oplifting{\phi_{m, m'\cdot C}}(m\cdot \opcleavage{(\phi_{m', C})}(m'\cdot X))$} (gmc);
    \end{tikzpicture}
  \end{center}
  Note that the coherence diagrams for $\phi$ and $\mu$ is obtained by applying $p$ to this whole diagram.

  The coherent natural transformations $\eta$ and $\mu$ defined above are isomorphisms, using ($*$), by several properties of opcartesian morphisms.
  The coherence conditions are 
  reduced to those for the action on $\E$ by using the universalities of $\oplifting{\phi}$s.
\end{proof}

\begin{proof}[Proof of Proposition \ref{prop:equivariant and fibration}]
  Let $\phi$ be the coherent natural transformation associated to $\Gamma$.
  This proposition can be proved in a similar way in Proposition \ref{prop:lax equivariant and opfibration}.
  The $\M$-action on $\G$ is defined by using the opcartesian lifting 
  of $\morphism{\inverse{\phi^{\Gamma}}}{\Gamma (m\cdot C)}{m\cdot\Gamma C}$.

  We present only the definition of the $\M$-action on $\G$.
  For any object $m$ in $\M$ and $(C, X)$ in $\G$, $m\cdot (C, X)$ is defined to be $(m\cdot C, \cleavage{(\inversew{\phi_{m, C}})}(m\cdot X))$.
  The following figure states that this definition is well-defined.
  \begin{center}
    \begin{tikzpicture}
      \matrix[row sep=2.5cm, column sep=3.5cm]{
        \node[left] (star) {$\cleavage{(\inversew{\phi_{m, C}})}(m\cdot X)$}; & \node[right] (mx) {$m\cdot X$};\\
        \node[left]  (gmc) {$\Gamma(m\cdot C)$}; & \node[right] (mgc) {$m\cdot \Gamma C$};\\
      };
      \node [below=.3cm of mgc] (mpx) {$m\cdot pX$};
      \node [below=.3cm of mpx] (pmx) {$p(m\cdot X)$};
      \draw 
      (star) edge node [above] {$\lifting{(\inversew{\phi_{m, C}})}(m\cdot X)$} (mx)
        (gmc) edge node [above] {$\inversew{\phi_{m, C}}$} (mgc)
        (mgc) edge [equal] (mpx)
        (mpx) edge [equal] (pmx);

      \draw[->, decorate, decoration=snake] (.7, -.5) -- (.7, .5);
      \node at (2, .0) [scale=1, align=center, inner sep=.5mm] {cartesian\\lifting};
    \end{tikzpicture}
  \end{center}

  Because $\phi_{m, C}$ is an isomorphism, $\lifting{\phi_{m, C}}$ is also an isomorphism, and thus $m'\cdot \lifting{\phi_{m, C}}$ is also an isomophism.
  Theorefore, $m'\cdot \lifting{\phi_{m, C}}$ is also cartesian because any isomorphism is cartesian.
  This is why the condition $(*)$ can be omitted.
\end{proof}

\begin{proof}[Proof of Proposition \ref{prop:model by morphism of model}]
  By proposition \ref{prop:simple gluing}.
  We define $L$ as follows.
  \begin{align*}
    Lv &= (F'v, Fv, \morphism{\theta_{v}}{F'v}{HFv})\\
    L(\morphism{a}{v}{u}) &= \morphism{(F'a, Fa)}{\theta_{v}}{\theta_{u}}
  \end{align*}
  where $v$ is any object in $\V$ and $\morphism{a}{v}{u}$ is any morphism in $\V$.
  It follows from the naturality of $\theta$ that $L\alpha$ above is a morphism in $\glue{\C'}{H}$.
  \begin{center}
    \begin{tikzpicture}
      \matrix[row sep=1.5cm, column sep=1.5cm]{
        \node (F'v) {$F'v$}; & \node (F'u) {$F'u$};\\
        \node (HFv) {$HFv$}; & \node (HFu) {$HFu$};\\
      };
      \draw
        (F'v) edge node [above] {$F'\alpha$} (F'u)
              edge node [left] {$\theta_v$} (HFv)
        (F'u) edge node [right] {$\theta_u$} (HFu)
        (HFv) edge node [below] {$HF\alpha$} (HFu);
    \end{tikzpicture}
  \end{center}

  Then, we define $\phi^{L}$ to be $(\phi^{F'}, \phi^F)$. The component $\phi^{L}_{v, u}=(\phi^{F'}_{v, u}, \phi^{F}_{v, u})$ is indeed
  a morphism $\morphism{}{v\cdot \theta_{u}}{\theta_{v\cdot u}}$ in $\glue{\C'}{H}$ by the definition \ref{def:morphism of models}.

  The coherent natural transformation $\phi^{L}$ is an isomorphism and satisfies the coherence since so $m_F$ and $m_{F'}$ are and do, and
  the composition in $\glue{\C'}{H}$ is defined using that in $\C'$ in a componentwise way.
\end{proof}

\begin{proof}[Proof of Proposition \ref{prop:equivariant fibration from monoidal}]
  We first define the functor $p_{T}$ by
  \begin{align*}
    p_T(X) &\defeq pX \\
    p_T(\morphism{f}{X}{\tilde TY}) &\defeq \morphism{pf}{pX}{TpY}.
  \end{align*}
  Notice $Tp=p\tilde T$ holds because $(T,\tilde T)$ is a monad over $p$.
  The functoriality of $p_T$ follows from $p\tilde\mu=\mu_{p}$ and $p\tilde \eta=\eta_{p}$; for example,
  $p_T$ perserves identities by the latter equation.
  For the $p_T$-cartesian lifting of $\fun{u}{I}{TpY}$, we can take $p$-cartesian lifting $\morphism{\lifting{u}}{u^*(\tilde TY)}{\tilde TY}$ of $u$. This is indeed $p_T$-cartesian by that fact that $\tilde T$ is fibred and $\tilde \mu$ is $p$-cartesian.

  Next, consider $\B\times_{\B_T}\E_{\tilde T}$. 
  Its object is a pair $(K, X)$ satisfying $K=pX$ and its morphism is a 
  pair $(u, f)$ satisfying $\eta_{p} \comp u=pf$. 
  Because $\eta_p=p\tilde\eta$ holds and $\tilde\eta$ is cartesian, 
  for each $f$ there exists a unique $h$ satisfying $\tilde \eta\comp h=f$ and $ph=u$.
  Using these facts, the functor $\morphism{F}{\B\times_{\B_T}\E_{\tilde T}}{\E}$
  defined as follows gives the isormophism $\B\times_{\B_T}\E_{\tilde T}\cong \E$ holds;
  \begin{align*}
    F(X) &= (pX, X)\\
    F(f) &= (pf, \tilde \eta\comp f).
  \end{align*}
  In addition, $F$ makes the following triangles commute.
  \begin{center}
    \begin{tikzpicture}
      \matrix[row sep=1cm, column sep=1cm]{
        & \node (e){$\E$};\\
      \node (b) {$\B$}; & \node (be) {$\B\times_{\B_T} \E_{\tilde T}$}; & \node (et) {$\E_{\tilde T}$};\\
      };
      \draw
        (e) edge node [above left] {$p$} (b)
            edge node [left] {$F$} node [right] {$\cong$} (be)
            edge node [above right] {$J$} (et)
        (be) edge (et)
        (be) edge (b)
        ;
    \end{tikzpicture}
  \end{center}

  Finally, we move on to the latter part of the proposition. 
  At first, notice that the tensor in $\E$ can be extended to the $\E$-action on $\E_{\tilde T}$ by
  using the strength $\tilde t$ of $\tilde T$. In a similar way, $\B_{T}$ has an $\E$-action by
  $X\cdot K=pX\tensor_{\B} K$ and 
  \[
    (\morphism{f}{X}{Y})\cdot (\morphism{g}{K}{TL})= \bigl((pX\tensor_{\B}K)\xrightarrow{pf\tensor_{\B}g}pY\tensor_{\B}TK\xrightarrow{t}T(pX\tensor_{\B}L)\bigr).
  \]
  It is straightforward to show that $p_T$ preserves this action strictly.
\end{proof}

The following proposition is claimed in the text right after Proposition \ref{prop:equivariant and fibration}.

\begin{proposition}[Constructions of Proposition \ref{prop:lax equivariant and opfibration} and \ref{prop:equivariant and fibration} coincide]
  Suppose $p$ is an opfibration in addition to the assumption in Proposition \ref{prop:equivariant and fibration}.
  The category $\G$ for pullback has two kinds of $\M$-actions by 
  Proposition \ref{prop:lax equivariant and opfibration} and \ref{prop:equivariant and fibration}.
  To distinguish these, we write $\G$ and $\G'$ for them.

  There are equivariant functors $\functor{H}{\G}{\G'}$ and $\functor{L}{\G'}{\G}$ such that
  $FG$ and $GF$ are identity functors as equivariant functors.
  In other words, $FG$ and $GF$ are identity functors
  in the (2-)category of $\M$-actegories and equivariant functors.
\end{proposition}

\begin{proof}
  Notice that $\G$ and $\G'$ are the same as categories but different as $\M$-actegories.
  Therefore, $H$ and $L$ may be identity functors. 
  The coherent natural isomorphisms for $H$ and $L$ are defined straightforwardly.
\end{proof}

\begin{proof}[Proof of non/fibredness of $T^{\top\top}$]
  At first, recall that a morphism in $\Sub(\Sets)$ is cartesian if and only if the corresponding square is pullback.
  \begin{itemize}
    \item fibredness of exception\par
      When $TX=X\uplus E$ and $(S\subseteq TR)$ is $(1\subseteq 1\uplus E)$, 
      $T^{\top\top}(P\subseteq X)={P\subseteq X\uplus E}$ holds for any $(P\subseteq X)$.
      Consider the pullback square on the left-hand side. The upper left corner $f^*(Q)$ is the inverse image of $Q$ by $f$.
      The square on the right-hand is also pullback because $(f\uplus E)^*(Q)=f^*(Q)$ holds.
      Therefore (the underlying functor of) $T^{\top\top}$ is fibred.
      \begin{center}
        \begin{tikzpicture}
          \matrix[row sep=.2cm, column sep=1cm]{
            \node(fq){$f^*(Q)$}; & \node (q) {$Q$}; & & \node (tfq) {$f^*(Q)$}; & \node (tq) {$Q$};\\
                                  \\& & \node{$\xmapsto{T^{\top\top}}$};\\
            \node(x){$X$}; & \node (y) {$Y$}; & & \node (tx) {$X\uplus E$}; & \node (ty) {$Y\uplus E$};\\
          };
          \draw 
            (fq) edge (q)
                  edge [inclusion] (x)
                  edge [pullback] (y)
            (x)  edge node [below] {$f$} (y)
            (q)  edge [inclusion] (y)
            (tfq) edge (tq)
                  edge [inclusion] (tx)
                  edge [pullback] (ty)
            (tx)  edge node [below] {$f\uplus E$} (ty)
            (tq)  edge [inclusion] (ty)
            ;
        \end{tikzpicture}
      \end{center}
      It is easy to see that the unit and multiplication are cartesian.
      \[
        \begin{tikzcd}
          P \ar[r] \ar[d, hookrightarrow] \ar[rd, phantom, very near start, "\lrcorner"] & P \ar[d, hookrightarrow] \\
          X \ar[r, "\eta"'] & X \uplus  E
        \end{tikzcd}
        \qquad
        \begin{tikzcd}
          P \ar[r] \ar[d, hookrightarrow] \ar[rd, phantom, very near start, "\lrcorner"] & P \ar[d, hookrightarrow] \\
          X \uplus E \uplus E \ar[r, "\mu"'] & X \uplus E
        \end{tikzcd}
      \]

  \item fibredness of nondeterminism\par
    When $(S\subseteq TR)$ is $(1 \subseteq 2 = \powersetfin(1))$, we can calculate: $T^{\top\top}(P\subseteq X)=\{ m \in \powersetfin(X) \mid \exists x \in m. P(x) \}$.
    For $X$ and $Q \subseteq Y$, consider the following pullback square.
    \begin{center}
      \begin{tikzcd}
        \mathbb{P} \ar[r] \ar[d, hookrightarrow] \ar[rd, phantom, very near start, "\lrcorner"] & T^{\top\top}(Q) \ar[d, hookrightarrow] \\
        \powersetfin(X) \ar[r, "\powersetfin(f)"'] & \powersetfin(Y)
      \end{tikzcd}
    \end{center}
    We can confirm that $\mathbb{P}$ is given by $T^{\top\top}(f^\ast(Q))$ by the following calculation.
    \begin{align*}
      m \in \mathbb{P} &\iff \exists y \in \powersetfin(f)(m). Q(y) \\
      &\iff \exists x \in m. Q(f(x)) \\
      &\iff \exists x \in m. x \in f^\ast(Q) \\
      &\iff m \in T^{\top\top}(f^\ast(Q))
    \end{align*}
    Therefore, the underlying functor of $T^{\top\top}$ preserves cartesian morphisms.
    The unit is cartesian because $\eta(x) = \{x\} \in T^{\top\top}(P) = \{ m \in \powersetfin(X) \mid \exists x \in m. P(x) \}$ if and only if $x \in P$ holds.
    To see that the multiplication is cartesian, again cosider a pullback diagram as follows.
    \[
      \begin{tikzcd}
        \mathbb{P} \ar[r] \ar[d, hookrightarrow] \ar[rd, phantom, very near start, "\lrcorner"] & T^{\top\top}(P) \ar[d, hookrightarrow] \\
        \powersetfin(\powersetfin(X)) \ar[r, "\mu"'] & \powersetfin(X)
      \end{tikzcd}
    \]
    Let us calculate $\mathbb{P}$:
    \begin{align*}
      M \in \mathbb{P} &\iff \mu(M) \in T^{\top\top}(P) \\
      &\iff \exists x \in \mu(M). P(x) \\
      &\iff \exists x \in \{ x \in X \mid \exists m \in M. x \in m \}. P(x) \\
      &\iff \exists m \in M. \exists x \in m. P(x) \\
      &\iff \exists m \in M. m \in T^{\top\top}(P) \\
      &\iff M \in T^{\top\top}(T^{\top\top}(P)).
    \end{align*}
    Therefore, the multiplication is cartesian.

  \item non-fibredness of side-effect\par
    For the side-effect monad $T(X) \defeq (A \times X)^A$ and $(S \subseteq T(R)) \defeq (B \tilde{\Rightarrow} B \tilde{\times} 1 \subseteq T1)$ for some $\emptyset \subsetneq B \subsetneq A$, $T^{\top\top}(P)$ is given by $B \tilde{\Rightarrow} B \tilde{\times} P$, where for $(P \subseteq X)$ and $(Q \subseteq Y)$, $P \tilde{\Rightarrow} Q \defeq \{ f \colon X \to Y \mid \forall x \in X. P(x) \Rightarrow Q(f(x)) \}$ and $P \tilde{\times} Q \defeq \{ (x, y) \in X \times Y \mid P(x) \land Q(y) \}$.
    Let $A = 2 = \{ 0, 1 \}$ and $B = 1 = \{ 0 \}$.
    Consider the pullback of the multiplication.
    \[
      \begin{tikzcd}
        \mathbb{P} \ar[r] \ar[d, hookrightarrow] \ar[rd, phantom, very near start, "\lrcorner"] & 1 \tilde{\Rightarrow} 1 \tilde{\times} P \ar[d, hookrightarrow] \\
        (2 \times (2 \times X)^2)^2 \ar[r, "\mu"'] & (2 \times X)^2
      \end{tikzcd}
    \]
    We will see that $\mathbb{P}$ is not $T^{\top\top}(T^{\top\top}(P))$.
    We fix $(P \subseteq X) \defeq (1 \subseteq 1)$.
    Let $f$ be an element in $(2 \times (2 \times X)^2)^2$ defined as follows.
    \[
      f(s) \defeq (1, \lambda s'. (0, 0))
    \]
    $\mu$ sends this $f$ to $\lambda s. (0, 0)$.
    Clearly, $\mu(f)$ is in $1 \tilde{\Rightarrow} 1 \tilde{\times} P$.
    Therefore, $f$ is in $\mathbb{P}$.
    However $f$ is not in $T^{\top\top}(T^{\top\top}(P))$, which concludes that $T^{\top\top}(T^{\top\top}(P)) \not= \mathbb{P}$.
    (Here we assume that $\mathbb{P} \subseteq (2 \times (2 \times X)^2)^2$ w.l.o.g.)

  \item non-fibredness of continuation\par
    When $TX=2^{2^X}$, $R=\emptyset$ and $(S\subseteq 2^{2^\emptyset})=(1\subseteq 2)$,
    \[
      T(P\subseteq X)=\{m\subseteq \P(X)\mid \forall P\subseteq C.\ C\in m\}
    \]
    holds for any $(P\subseteq X)$.

    In the sequel, we identify $2^{2^{\N}}$ with the set of sets of real numbers in $[0,1]$.
    Consider the left pullback square where $\morphism{0}{\N}{\N}$ is a constant map sending every number to $0$.
    and apply $T^{\top\top}$ to the diagram.
    \begin{center}
      \begin{tikzpicture}
        \matrix[row sep=.1cm, column sep=1cm]{
          \node(phi1) {$\emptyset$}; & \node (phi2){$\emptyset$}; & & \node (tp1){$\tilde T\emptyset$}; & \node (tp2){$\tilde T\emptyset$};\\
                                      & & \node {$\xmapsto{T^{\top\top}}$};\\
          \node(n1) {$\N$}; & \node (n2){$\N$}; & & \node (pi1){$\P([0,1])$}; & \node (pi2){$\P([0,1])$};\\
        };
        \draw
          (phi1) edge (phi2)
                  edge [inclusion] node [left] {$\bang$} (n1)
                  edge [pullback] (n2)
          (phi2) edge [inclusion] node [right] {$\bang$} (n2)
          (n1)   edge node [below] {$0$} (n2)
          (tp1) edge (tp2)
                  edge [inclusion] (pi1)
          (tp2) edge [inclusion]  (pi2)
          (pi1)   edge node [below] {$T0$} (pi2)
          ;
      \end{tikzpicture}
    \end{center}
    In the right square, the nodes and edges are determined as follows:
    \begin{align*}
      T(0)(U) & =
        \begin{cases}
          [0,1] & (\text{$0\in U$ and $1\in U$})\\
          [0,1/2) & (\text{$0\in U$ and $1\notin U$})\\
          [1/2,1] & (\text{$0\notin U$ and $1\in U$})\\
          \emptyset & (\text{$0\notin U$ and $1\in U$})
        \end{cases}\\
      T(\emptyset\subseteq \N)&=\{m\subseteq [0,1]\mid \forall r.\ r\in m\}\\
                              &= \{[0,1]\}.
    \end{align*}
    Then, we get 
    \[
      T^{\top\top}(\emptyset)\times_{\P([0,1])}\P([0,1]) = \{U\subseteq \text{$0\in U$ and $1\in U$}\}\supsetneq \{[0,1]\} = T^{\top\top}(\emptyset)
    \]
    and therefore $T^{\top\top}$ is not fibred.
  \end{itemize}
\end{proof}

\end{document}